\pdfoutput=1
\documentclass[journal,twocolumn]{IEEEtran}
\IEEEoverridecommandlockouts

\usepackage{amsmath,amssymb,amsfonts,amsthm,bm}
\usepackage{algorithm}
\usepackage{algorithmic}
\usepackage{balance}
\usepackage[font=footnotesize,skip=1pt]{caption}
\usepackage{cite}
\usepackage{enumitem}
\usepackage{graphicx}
\usepackage{mathtools}
\usepackage[protrusion=true,expansion=true]{microtype}
\usepackage{multirow}
\usepackage{subcaption}
\usepackage{units}
\usepackage{xcolor}

\newtheorem{proposition}{Proposition}

\allowdisplaybreaks
\DeclareMathOperator*{\argmin}{arg\,min}
\DeclareMathOperator*{\argmax}{arg\,max}

\graphicspath{{Pictures_Paper/}}

\begin{document}

\title{Dynamic and Robust Sensor Selection Strategies for Wireless Positioning with TOA/RSS Measurement}

\author{
Myeung~Suk~Oh,~\IEEEmembership{Student~Member,~IEEE}, Seyyedali~Hosseinalipour,~\IEEEmembership{Member,~IEEE}, Taejoon~Kim,~\IEEEmembership{Senior~Member,~IEEE}, David~J.~Love,~\IEEEmembership{Fellow,~IEEE},\\ James~V.~Krogmeier,~\IEEEmembership{Senior~Member,~IEEE}, and Christopher~G.~Brinton,~\IEEEmembership{Senior~Member,~IEEE}
\thanks{This work was supported in part by Ford Motor Company and
in part by the National Science Foundation (NSF) under Grants CNS2146171, CNS2212565, and CNS2225577. The authors thank R. Burke and D. Upadhyay for their valuable discussions.}
\thanks{M. S. Oh, D. J. Love, J. V. Krogmeier, and C. G. Brinton are with the School of Electrical and Computer Engineering, Purdue University, West Lafayette, IN, 47907 USA (e-mail: \{oh223, djlove, jvk, cgb\}@purdue.edu).}
\thanks{S. Hosseinalipour is with the Department of Electrical Engineering, University at Buffalo-SUNY, NY, 14260 USA (email: alipour@buffalo.edu).}
\thanks{T. Kim is with the Department of Electrical
Engineering and Computer Science, University of Kansas, KS, 66045 USA (email: taejoonkim@ku.edu).}
}

\maketitle

\begin{abstract}
    Emerging wireless applications are requiring ever more accurate location-positioning from sensor measurements.
    In this article, we develop sensor selection strategies for 3D wireless positioning based on time of arrival (TOA) and received signal strength (RSS) measurements to handle two distinct scenarios: (i) known approximated target location, for which we conduct dynamic sensor selection to minimize the positioning error; and (ii) unknown approximated target location, in which the worst-case positioning error is minimized via robust sensor selection.
    We derive expressions for the Cram\'er-Rao lower bound (CRLB) as a performance metric to quantify the positioning accuracy resulted from selected sensors.
    For dynamic sensor selection, two greedy selection strategies are proposed, each of which exploits properties revealed in the derived CRLB expressions.
    These selection strategies are shown to strike an efficient balance between computational complexity and performance suboptimality.
    For robust sensor selection, we show that the conventional convex relaxation approach leads to instability, and then develop three algorithms based on (i) iterative convex optimization (ICO), (ii) difference of convex functions programming (DCP), and (iii) discrete monotonic optimization (DMO).
    Each of these strategies exhibits a different tradeoff between computational complexity and optimality guarantee.
    Simulation results show that the proposed sensor selection strategies provide significant improvements in terms of accuracy and/or complexity compared to existing sensor selection methods.
\end{abstract}

\begin{IEEEkeywords}
    Cram\'er-Rao lower bound (CRLB), received signal strength, sensor selection, time of arrival, wireless positioning 
\end{IEEEkeywords}

\vspace{-5mm}
\section{Introduction}\label{sec:intro}

\subsection{Wireless Positioning and Sensor Selection}\label{ssec:intro_background}

Wireless positioning is employed in many applications across the military~\cite{Chong03,Rantakokko11} and commercial \cite{Sayed05,Ko21} sectors, e.g., for target tracking, system security, and smart automation~\cite{Gezici09}.
As the efficacy of these applications depends on the accuracy of location information and the speed at which it can be obtained, it is critical to maximize the performance of wireless positioning while minimizing the associated algorithmic complexity.
In wireless positioning, a target location is often estimated using a collected set of location-dependent measurements (e.g., received signal strength (RSS)) acquired from multiple geo-distributed sensors~\cite{Gustafsson05}.
In particular, with the advent of ultra-wideband (UWB) communications with large bandwidth signals, wireless positioning with time-sensitive measurements (e.g., time of arrival (TOA) and time difference of arrival (TDOA)) has become an active field of research~\cite{Gezici05,Alarifi16,Mazhar17}.

Although it has been analytically demonstrated that using more sensors results in improved positioning accuracy regardless of geographical placement~\cite{Patwari03}, deploying a large number of sensors is undesirable in practice.
For example, TOA-based positioning in general requires sensors to make isolated measurements~\cite{UWB20}, which may result in extensively prolonged positioning times with a large number of sensors.
Other practical considerations, like cost and packaging constraints, also limit the number of sensors that can be deployed for positioning.
As a result, in many settings, only a portion/subset of the placed sensors are actually \emph{selected} for usage.
Hence, selecting the most efficient group of sensors for optimal wireless positioning is a critical and yet challenging task~\cite{Rowaihy07}. 

In this article, we study the problem of optimal sensor selection for wireless positioning considering two distinct scenarios: (i) \emph{an approximate target location is known via prior prediction} and (ii) \emph{target location is not approximated and hence unknown}.
When the approximated target location is given, the sensors should be selected such that the positioning accuracy for the respective location is maximized.
We refer to this scenario as \emph{dynamic} sensor selection because the optimal set of sensors to be selected varies by the given approximated location.
On the other hand, when no information on the target location is available, the selection of sensors should be carried out to maximize the worst-case positioning accuracy.
As a result, we refer to this scenario as \emph{robust} sensor selection.

In each scenario, we mathematically formulate an optimization problem for sensor selection.
In doing so, we adopt the Cram\'er-Rao lower bound (CRLB)~\cite{Kay97} as a performance metric for quantifying the positioning accuracy.
In our work, the CRLB quantifies the lowest mean squared error (MSE) achievable from a selected set of sensors, optimization of which provides a natural solution to our sensor selection problem.
Based on the formulated optimization problems, we present novel sensor selection strategies that improve performance in accuracy and/or complexity of wireless positioning.

\subsection{Related Work}\label{ssec:intro_work}

To perform effective sensor selection, it is critical to understand the impact of sensor placement on the accuracy of wireless positioning~\cite{Patwari03,Qi02,Catovic04,Laaraiedh12,Li18,Li19}.
In~\cite{Patwari03}, the CRLBs on wireless positioning were separately derived for line-of-sight (LOS) TOA and RSS measurements, and it was shown that the positioning accuracy strongly depends on the \emph{geometric conditioning} (i.e., the geometric arrangement with respect to a target) of the sensors.
The CRLB on TOA-based wireless positioning over both LOS and non-LOS (NLOS) channels was derived in~\cite{Qi02}, revealing that the bound is strictly a function of the LOS channel unless the bias in the NLOS channel is compensated.
Recent works, e.g.,~\cite{Catovic04,Li18,Laaraiedh12,Li19}, have extended the analysis and derived the CRLB upon using hybrid measurements.
They showed that the target positioning accuracy can be improved via joint consideration of different measurement types.
In addition to improving the theoretical bound, the benefit of utilizing multiple measurement types has been verified in state-of-the-art positioning schemes.
For example, a picocell-based joint TOA and direction of arrival (DOA) estimation is proposed in~\cite{Pan22}, and a fingerprint localization with RSS and channel state information (CSI) measurement is considered in~\cite{Zhou21}.

Sensor selection problems are in general NP-hard combinatorial optimizations, the exact solutions to which can be obtained via exhaustive search.
Since exhaustive search strategies suffer from prohibitive computation burden as the number of sensors increases, some works rely on convex relaxations and/or heuristic strategies to find computationally efficient suboptimal solutions~\cite{Godrich12,Zhao19,Dai20}.
Specifically, in~\cite{Godrich12}, sensor selection was formulated as a knapsack problem with its solution obtained via greedy algorithms that aim to minimize the RSS-driven CRLB.
Both semidefinite relaxation (SDR) and heuristic methods for sensor selection aiming to minimize the CRLB in TDOA-based wireless positioning were proposed in~\cite{Zhao19}.
The authors in~\cite{Dai20} solved a convex-relaxed CRLB-minimizing sensor selection problem via semidefinite programming with randomization for both TDOA-based and TOA-based wireless positioning. 

Although these prior works provide useful insights on sensor selection problems, their approaches rely on complete evaluation of the CRLB.
Even with greedy selection methods~\cite{Godrich12,Zhao19}, the computational load from evaluating the bound, which requires a matrix inverse operation of complexity $\mathcal{O}(n^3)$, can be burdensome if a system involves a large number of sensors.
Also, as all of these methods utilize the CRLB for non-Bayesian estimation (i.e., the target location is assumed to be known a priori), they are only applicable in scenarios where precise information on the target location is initially available.
In this work, we first investigate dynamic sensor selection~\cite{Godrich12,Zhao19,Dai20}, and propose low-complexity sensor selection strategies that avoid computing the entire CRLB expression.
Also, we take one step further from the conventional sensor selection literature and consider robust sensor selection, for which we focus on minimizing the worst-case CRLB and propose sensor selection strategies for unknown target locations.

\subsection{Outline and Summary of Contributions}\label{ssec:intro_outline}

This article focuses on sensor selection for wireless positioning with hybrid TOA/RSS measurements.
We consider both \emph{dynamic} and \emph{robust} sensor selection scenarios, and our major contributions can be summarized as follows:
\begin{itemize}[leftmargin=4mm]
    \item We derive the CRLB expression into two different forms, namely (i) \emph{trace} and (ii) \emph{fractional} forms.
    We use these expressions to formulate optimization problems and develop sensor selection strategies.
    Based on both forms derived, we show that the CRLB can be optimized without evaluating the entire expression, which we exploit to reduce the complexity of our dynamic sensor selection algorithms.
    \item We develop low-complexity greedy selection strategies for dynamic sensor selection.
    Based on our computation-efficient metrics, the proposed strategies perform a sequential sensor selection where in each iteration, the sensor minimizing the current CRLB is selected, one at a time.
    Our numerical results demonstrate that, compared to the benchmarks, the proposed strategies provide comparable positioning performance with much less complexity.
    \item We propose and study the robust sensor selection problem, 
    where we reveal that the conventional convex optimization approach provides unreliable binary solutions.
    We subsequently propose three sensor selection strategies based on (i) iterative convex optimization (ICO), (ii) difference of convex functions programming (DCP), and (iii) discrete monotonic optimization (DMO), each of which has a different tradeoff between complexity and optimality guarantee.
    Our numerical results show that each finds solutions that are stable and effective in the worst-case CRLB minimization.    
\end{itemize}

The rest of this article is organized as follows.
Section~\ref{sec:system} describes our wireless positioning system and preliminaries on sensor selection.
Two different forms of the CRLB expression, which are utilized in our sensor selection strategies, are derived in Section~\ref{sec:CRLB}.
In Section~\ref{sec:dynamic}, two greedy algorithms based on unique selection metrics are proposed for the dynamic sensor selection problem.
In Section~\ref{sec:robust}, three different strategies are developed to address the robust sensor selection problem.
Simulation results are presented in Section~\ref{sec:simulations}, and Section~\ref{sec:conclusion} concludes this article.

\vspace{-1.5mm}
\section{System Model and Preliminaries}\label{sec:system}

We first describe our system configuration in Section~\ref{ssec:system_layout}.
Then, our hybrid TOA/RSS measurement model is introduced in Section~\ref{ssec:system_measurement}.
The data collection and estimation steps are explained in Section~\ref{ssec:system_collection}.
Finally, our dynamic and robust sensor selection problems are formulated in Section~\ref{ssec:system_problem}.

\begin{figure}[!t]
    \centering
    \includegraphics[width=0.65\linewidth]{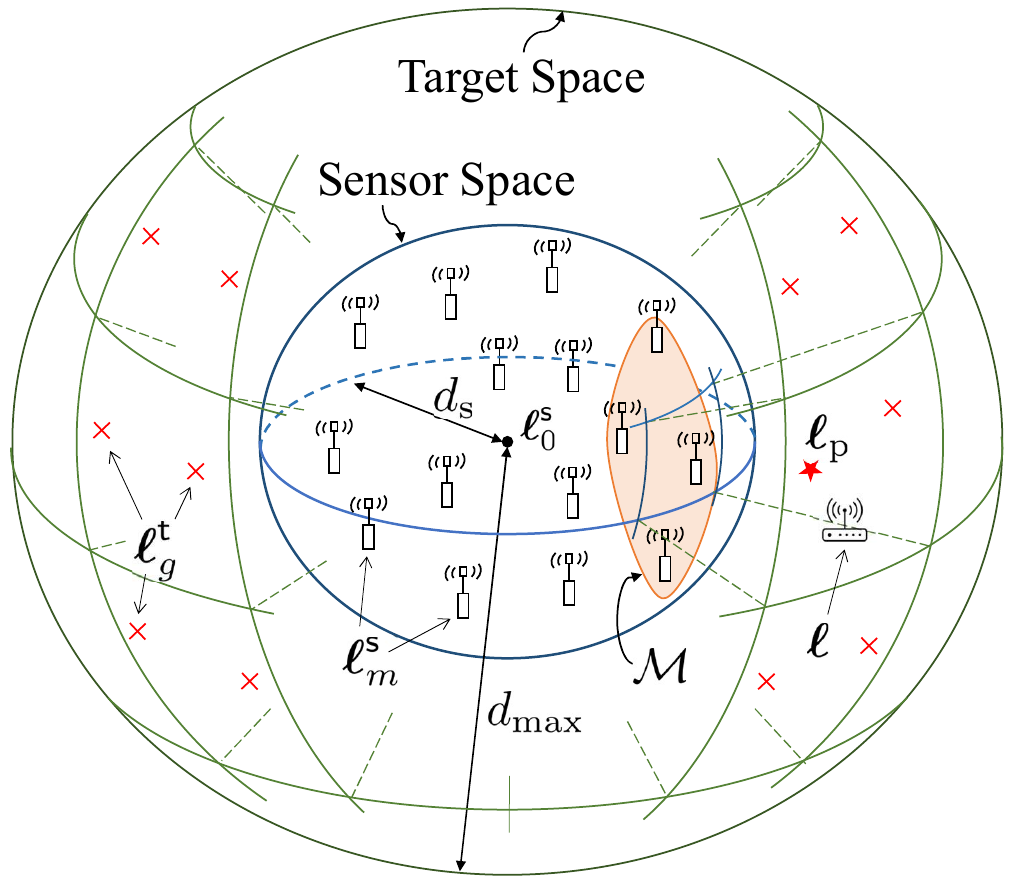}
    \caption{Geographical configuration of the sensors and candidate target locations in our wireless positioning model.}
    \vspace{-2mm}
    \label{fig:system_layout}
\end{figure}

\vspace{-2.5mm}
\subsection{Sensor Geography Model}\label{ssec:system_layout}

As shown in Fig.~\ref{fig:system_layout}, we consider 3D wireless positioning with sensors from a set $\mathcal{M}_\text{max}=\{1,\ldots,M_\text{max}\}$, where $M_\text{max}$ is set by practical system limitations.
Each sensor $m\in\mathcal{M}_\text{max}$ is placed at a predetermined 3D location $\boldsymbol{\ell}^\mathsf{s}_m=[x^\mathsf{s}_m, y^\mathsf{s}_m, z^\mathsf{s}_m]^\top$.
We denote the center point of these sensors as $\boldsymbol{\ell}^\textsf{s}_0=\frac{1}{M_\text{max}}\sum_{m=1}^{M_\text{max}}\boldsymbol{\ell}^\mathsf{s}_m$ and characterize the space in which the sensors are located via a constant $d_\text{s}$, the smallest positive value satisfying $\|\boldsymbol{\ell}^\mathsf{s}_m-\boldsymbol{\ell}^\textsf{s}_0\|_2\leq d_\text{s}$, $\forall m \in \mathcal{M}_\text{max}$.

We consider single-antenna sensors with no time-synchronization, i.e., a hardware-limited positioning system where sensors cannot jointly collect their measurements and conduct parameter estimation.
We aim to conduct wireless positioning on a single stationary target with location $\boldsymbol{\ell}=[x,y,z]^\top$ which is outside of the sensor space, i.e., $d_\text{s}<\|\boldsymbol{\ell}-\boldsymbol{\ell}^\textsf{s}_0\|_2\leq d_\text{max}$, where $d_\text{max}$ is the maximum distance to which wireless positioning can be conducted.
We divide our target space into $G$ distinct regions and represent each of them via a representative point (e.g., the center of the region) $\boldsymbol{\ell}^\mathsf{t}_g=[x^\mathsf{t}_g,y^\mathsf{t}_g,z^\mathsf{t}_g]^\top$ that satisfies $d_\text{s}<\|\boldsymbol{\ell}^\mathsf{t}_g-\boldsymbol{\ell}^\textsf{s}_0\|_2\leq d_\text{max}$, $\forall g=1,2,\ldots,G$, as shown in Fig.~\ref{fig:system_layout}.
We define $\mathcal{L}=\{\boldsymbol{\ell}^\mathsf{t}_1,\ldots,\boldsymbol{\ell}^\mathsf{t}_G\}$ as a  set collecting these $G$ representative points/locations.

For sensor selection, we assume that $M$ sensors must be selected from the $M_\text{max}$ total. 
We define a set $\mathcal{M}\subseteq\mathcal{M}_\text{max}$, with size $M=\vert\mathcal{M}\vert\leq M_\text{max}$, to be the index set of the $M$ sensors selected for positioning.
We also define $\boldsymbol{m}_\mathcal{M}=[m_1,m_2,\ldots,m_M]^\top$ to be a vector listing the elements of $\mathcal{M}$ in ascending order.
Once the selection is made, the system executes a sequence of steps to conduct wireless positioning.
The overall procedure of wireless positioning using $M$ sensors selected from $\mathcal{M}_\text{max}$ is illustrated in Fig.~\ref{fig:system_model}, and details on each step are provided in the following sections.

\begin{table*}[t]
\captionsetup{justification=centering, labelsep=newline}
\centering
\caption{A list of variables describing our system model}
\label{tb:variables}
\begin{tabular}{|c|l||c|l|} 
    \hline
    Variable & Description & Variable & Description \tabularnewline
    \hline
    $\mathcal{M}_\text{max}$ & Set of entire sensors & $h_m$ & Channel gain between the target and sensor $m$ \tabularnewline
    \hline
    $\mathcal{M}$ & Set of sensors selected for positioning & $w_m(t)$ & Noise on the received signal of sensor $m$ \tabularnewline
    \hline 
    $\boldsymbol{\ell}^\mathsf{s}_m$ & 3D coordinates of sensor $m$ & $\xi$ & Pathloss exponent \tabularnewline
    \hline
    $\boldsymbol{\ell}^\textsf{s}_0$ & Center point of the sensors & $\widehat{\tau}_{m}$ ($\widehat{P}_m$) & TOA (RSS) measurement on sensor $m$ \tabularnewline
    \hline
    $\boldsymbol{\ell}$ & 3D coordinates of the target & $n_{\text{T},m}$ ($n_{\text{R},m}$) & Noise on TOA (RSS) measurement of sensor $m$ \tabularnewline
    \hline
    $\boldsymbol{\ell}^\mathsf{t}_g$ & Representative point for region $g$ & $P_{0,m}$ ($d_{0,m}$) & Reference power (distance) of sensor $m$ \tabularnewline
    \hline
    $d_\text{s}$ & Radius of the sensor space & $\widehat{d}_{\text{T},m}$ ($\widehat{d}_{\text{R},m}$) & Distance estimated from TOA (RSS) of sensor $m$ \tabularnewline
    \hline
    $d_\text{max}$ & Maximum distance for positioning & $e_{\text{T},m}$ ($e_{\text{R},m}$) & Error on TOA-based (RSS-based) distance estimation of sensor $m$ \tabularnewline
    \hline
    $d_m$ & Distance between the target and sensor $m$ & $\sigma^2_{\text{T},m}$ ($\sigma^2_{\text{R},m}$) & Variance of $e_{\text{T},m}$ ($e_{\text{R},m}$) \tabularnewline
    \hline
    $T_\text{p}$ & Entire duration of wireless positioning & $\rho_{m,m'}$ ($\eta_{m,m'}$) & Spatial (Hybrid) correlation coefficient between sensors $m$ and $m'$ \tabularnewline
    \hline
    $s(t)$ & Reference signal & $\boldsymbol{\ell}_\text{p}$ & Prior approximation of $\boldsymbol{\ell}$ \tabularnewline
    \hline
    $T_\text{s}$ & Length of reference signal & $\widehat{\boldsymbol{\ell}}_\mathcal{M}$ & Estimation of $\boldsymbol{\ell}$ using the sensor set $\mathcal{M}$ \tabularnewline
    \hline
    $r_m(t)$ & Received signal at sensor $m$ & $\sigma^2_\mathcal{M}(\boldsymbol{\ell})$ & The CRLB on $\boldsymbol{\ell}$ obtained using the sensor set $\mathcal{M}$ \tabularnewline
    \hline
\end{tabular}
\vspace{-2mm}
\end{table*}

\vspace{-2mm}
\subsection{Hybrid TOA/RSS Measurement Model}\label{ssec:system_measurement}

We assume that the entire positioning procedure has a duration $T_\text{p}$.
At each positioning round, a predetermined reference signal $s(t)$ of duration $T_\text{s}$ is transmitted from the target and received by the sensors.
If we define the distance between sensor $m$ and the target as $d_{m}=\|\boldsymbol{\ell}^\mathsf{s}_m-\boldsymbol{\ell}\|_2$ and assume LOS propagation from the target to each sensor~\cite{Qi02}, the received signal at sensor $m$ can be expressed as
\begin{equation}
    r_m(t)=h_m s\Big(t-{d_m}/{c}\Big) + w_m(t),
    \label{eq:received_signal}
\end{equation}
where $h_m$ is the channel gain between the target and sensor $m$ such that $\vert h_m\vert ^2 \propto {d_m^{-\xi}}$, with $\xi$ being the pathloss exponent, $w_m(t)$ is zero-mean Gaussian noise, and $c$ is the speed of light.
We consider $T_\text{p}$ to be long enough so that the reference signal can be received by every sensor within a single period of positioning procedure, i.e., $\max_{m\in\mathcal{M}_\text{max}}(T_\text{s} + \frac{d_m}{c}) \ll T_\text{p}$.

\begin{figure}[!t]
    \centering
    \includegraphics[width=1\linewidth]{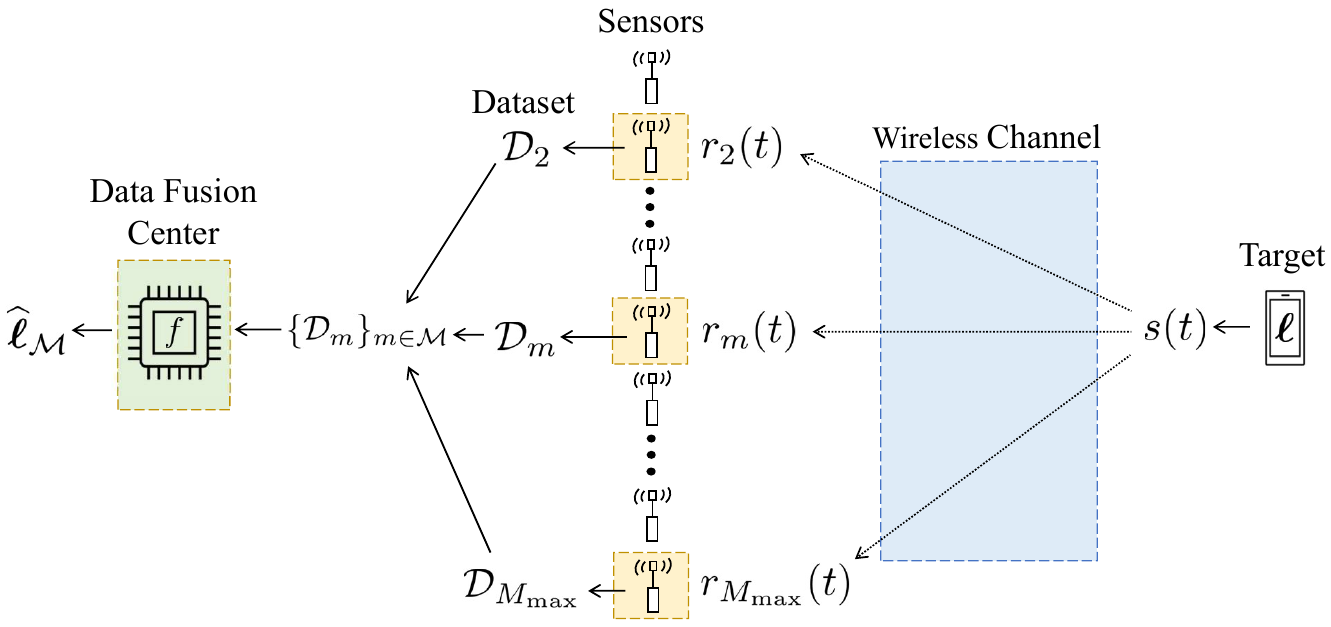}
    \caption{System model of 3D wireless positioning using multiple sensors. Yellow boxes indicate selected sensors.}
    \label{fig:system_model}
    \vspace{-2mm}
\end{figure}

We assume both TOA and RSS are measured by all sensors\footnote{Other potential parameters, e.g., TDOA and angle of arrival (AOA), are not available in our single-antenna and time-asynchronous sensor system model.}.
Each sensor $m\in \mathcal{M}$ measures its TOA $\widehat{\tau}_{m}$ and RSS~$\widehat{P}_{m}$ as
\begin{equation}
    \widehat{\tau}_{m} = {d_{m}}/{c} + n_{\text{T},m}
    \label{eq:estimated_TOA}
\end{equation}
and
\begin{equation}
    \widehat{P}_{m} = P_{0,m}-10\xi\log_{10}\left({d_{m}}/{d_{0,m}}\right) + n_{\text{R},m},
    \label{eq:estimated_RSS}    
\end{equation}
where $n_{\text{T},m}$ and $n_{\text{R},m}$ are zero-mean real Gaussian noises for TOA and RSS measurements at sensor $m$, respectively.
Noises arise due to environmental (e.g., channel fading and shadowing) and systematic (e.g., non-ideal correlator and packet latency) factors~\cite{Macii13}.
In~\eqref{eq:estimated_RSS}, $P_{0,m}$ and $d_{0,m}$ are the reference power and reference distance of sensor $m$, respectively.

We assume maximum likelihood estimation (MLE) for approximating $d_m$ from the acquired measurements.
The estimated distances $\widehat{d}_{\text{T},m}$ and $\widehat{d}_{\text{R},m}$ from the respective TOA and RSS measurements are then expressed as~\cite{Li18,Jia08,Shen12,Huang14}
\begin{equation}
    \widehat{d}_{\text{T},m} = d_{m} + e_{\text{T},m} \;\text{ and }\; \ln \widehat{d}_{\text{R},m} = \ln d_{m} + e_{\text{R},m}, \nonumber
    \label{eq:estimated_dist_TOA/RSS}
\end{equation}
where $e_{\text{T},m}$ and $e_{\text{R},m}$ are zero-mean real Gaussian error distributions with variances $\sigma^2_{\text{T},m}$ and $\sigma^2_{\text{R},m}$, respectively.
As a result, given the target location $\boldsymbol{\ell}$, the probability density functions (PDFs) of $\widehat{d}_{\text{T},m}$ and $\widehat{d}_{\text{R},m}$ are given by~\cite{Laaraiedh12}
\begin{equation}
    p_{\text{T},m}(\widehat{d}_{\text{T},m}\vert \boldsymbol{\ell})=\frac{1}{\sqrt{2\pi}\sigma_{\text{T},m}} e^{-\frac{(\widehat{d}_{\text{T},m}-d_m)^2}{2\sigma^2_{\text{T},m}}}
    \label{eq:TOA_PDF}
\end{equation}
and
\begin{equation}
    p_{\text{R},m}(\widehat{d}_{\text{R},m}\vert \boldsymbol{\ell}) = \frac{1}{\sqrt{2\pi}\widehat{d}_{\text{R},m}\sigma_{\text{R},m}}e^{-\frac{(\ln\widehat{d}_{\text{R},m}-\ln d_m)^2}{2\sigma_{\text{R},m}^2}},
    \label{eq:RSS_PDF}
\end{equation}
which are Gaussian and log-normal, respectively.

For sensor $m\in\mathcal{M}$, we define $\widehat{\boldsymbol{q}}_{m} = [\widehat{d}_{\text{T},m}, \ln\widehat{d}_{\text{R},m}]^\top$~and $\boldsymbol{q}_{m}\hspace{-0.5mm}=\hspace{-0.5mm}[d_{m}, \ln d_{m}]^\top$.
We further define $\widehat{\boldsymbol{q}}_{\mathcal{M}}\hspace{-0.5mm}=\hspace{-0.5mm} [\widehat{\boldsymbol{q}}_{m_1}^\top,\ldots,\widehat{\boldsymbol{q}}_{m_M}^\top]^\top$ and $\boldsymbol{q}_{\mathcal{M}} = [\boldsymbol{q}_{m_1}^\top,\ldots,\boldsymbol{q}_{m_M}^\top]^\top$ to be the vertical concatenations of $\widehat{\boldsymbol{q}}_{m}$ and $\boldsymbol{q}_{m}$, respectively, from the sensors $m \in \mathcal{M}$.
The joint PDF of $2M$ distance estimates from the hybrid TOA/RSS measurements across the sensor set $\mathcal{M}$ is then given by~\cite{Fletcher06}
\begin{align}
    p_{\text{H},\mathcal{M}}(\widehat{\boldsymbol{q}}_{\mathcal{M}}\vert \boldsymbol{\ell})
    =&\frac{e^{-\frac{1}{2}(\widehat{\boldsymbol{q}}_{\mathcal{M}}-\boldsymbol{q}_{\mathcal{M}})^\top\mathbf{R}_{\mathcal{M}}^{-1}(\widehat{\boldsymbol{q}}_{\mathcal{M}}-\boldsymbol{q}_{\mathcal{M}})}}{(2\pi)^{M}(\prod_{m\in\mathcal{M}}\widehat{d}_{\text{R},m})\det(\mathbf{R}_{\mathcal{M}})^{\frac{1}{2}}},
    \label{eq:joint_hybrid_PDF}
\end{align}
where $\mathbf{R}_{\mathcal{M}}$ is the $2M \times 2M$ covariance matrix of $\widehat{\boldsymbol{q}}_{\mathcal{M}}$, i.e., $\mathbf{R}_{\mathcal{M}}=\mathbb{E}[(\widehat{\boldsymbol{q}}_{\mathcal{M}}-\mathbb{E}[\widehat{\boldsymbol{q}}_{\mathcal{M}}])(\widehat{\boldsymbol{q}}_{\mathcal{M}}-\mathbb{E}[\widehat{\boldsymbol{q}}_{\mathcal{M}}])^\top]$.
Correlation among the measurement noises are captured by the entries of $\mathbf{R}_{\mathcal{M}}$, which we express as the following partitioned block matrix:
\begin{equation}
    \mathbf{R}_{\mathcal{M}}=
    \begin{bmatrix}
        \mathbf{R}_{m_1m_1} & \mathbf{R}_{m_1m_2} & \cdots & \mathbf{R}_{m_1m_M} \\
        \mathbf{R}_{m_2m_1} & \mathbf{R}_{m_2m_2} & \cdots & \mathbf{R}_{m_2m_M} \\
	\vdots & \vdots & \ddots & \vdots \\
        \mathbf{R}_{m_Mm_1} & \mathbf{R}_{m_Mm_2} & \cdots & \mathbf{R}_{m_Mm_M}
    \end{bmatrix},
    \label{eq:2Mx2M_cov_matrix}
\end{equation}
where $\mathbf{R}_{m_im_j}$ is the $2 \times 2$ covariance matrix between $\widehat{\boldsymbol{q}}_{m_i}$ and $\widehat{\boldsymbol{q}}_{m_j}$, $\forall i,j$.
The general expression for $\mathbf{R}_{m_im_j}$ is given by
\begin{equation}
    \mathbf{R}_{m_im_j}\hspace{-1mm}=\hspace{-0.5mm}\rho_{m_im_j}\hspace{-1mm}
    \begin{bmatrix}
        \sigma_{\text{T},m_i}\sigma_{\text{T},m_j} \hspace{-1.5mm}&\hspace{-1.5mm} \eta_{m_im_j}\sigma_{\text{T},m_i}\sigma_{\text{R},m_j} \\
        \eta_{m_im_j}\sigma_{\text{R},m_i}\sigma_{\text{T},m_j} \hspace{-1.5mm}&\hspace{-1.5mm} \sigma_{\text{R},m_i}\sigma_{\text{R},m_j}
    \end{bmatrix}
    \hspace{-0.5mm},\hspace{-1.5mm}
    \label{eq:2x2_cov_matrix}
\end{equation}
where $\rho_{m_im_j},\:\eta_{m_im_j}\hspace{-0.5mm}\in\hspace{-0.5mm}[0,1)$ are the spatial and hybrid corr-elation coefficients, respectively, between sensors $m_i$ and $m_j$.

\vspace{-3mm}
\subsection{Data Collection and Location Estimation}\label{ssec:system_collection}

Once sensor $m \in \mathcal{M}$ completes TOA and RSS measurements from the received signal $r_m(t)$, it generates the data point $\mathcal{D}_{m}=\{\widehat{\boldsymbol{q}}_{m}\}$ and transfers it to the central data fusion center in Fig.~\ref{fig:system_model}.
Subsequently, based on the $M$ data points collected from the sensors in $\mathcal{M}$, the data fusion center computes the estimated location of the target $\widehat{\boldsymbol{\ell}}_{\mathcal{M}} = [\widehat{x},\widehat{y},\widehat{z}]$ using a location estimation function $f_\text{est}$, i.e., $f_\text{est}: \{\mathcal{D}_{m}\}_{m\in \mathcal{M}}\rightarrow \widehat{\boldsymbol{\ell}}_{\mathcal{M}}$.
The function $f_\text{est}$ can be modeled using various localization algorithms, e.g., Taylor expansion~\cite{Foy76} or weighted least squares~\cite{Kay97}.
We make no specific assumption on $f_\text{est}$ so that our sensor selection strategies are compatible with any estimation method.

\vspace{-3mm}
\subsection{Problem Formulation}\label{ssec:system_problem}

We formally define the accuracy of wireless positioning using a selected sensor set $\mathcal{M}$ via MSE, which is given by
\begin{equation}
    \text{MSE}_\mathcal{M}(\boldsymbol{\ell}) = \mathbb{E}\big[(\widehat{x}-x)^2+(\widehat{y}-y)^2+(\widehat{z}-z)^2\big].
\end{equation}
If we define $\sigma^2_{\mathcal{M}}(\boldsymbol{\ell})$ to be the CRLB obtained via sensor set $\mathcal{M}$ for a target located at $\boldsymbol{\ell}$, $\sigma^2_{\mathcal{M}}(\boldsymbol{\ell})$ is a lower bound on $\text{MSE}_\mathcal{M}(\boldsymbol{\ell})$ satisfying the following relationship~\cite{Catovic04}:
\begin{equation*}
    \sigma^2_{\mathcal{M}}(\boldsymbol{\ell})\hspace{-0.5mm}\leq\hspace{-0.5mm}
    \mathbb{E}\big[\|\widehat{\boldsymbol{\ell}}_{\mathcal{M}}-\boldsymbol{\ell}\|^2_2\big] \hspace{-0.5mm}=\hspace{-0.5mm}\mathbb{E}\big[(\widehat{x}-\hspace{-0.5mm}x)^2\hspace{-0.5mm}+\hspace{-0.5mm}(\widehat{y}-\hspace{-0.5mm}y)^2\hspace{-0.5mm}+\hspace{-0.5mm}(\widehat{z}-\hspace{-0.5mm}z)^2\big].
\end{equation*}

Based on the CRLB, we formulate two distinct sensor selection problems:

\subsubsection{Dynamic sensor selection} When an approximation of the target location $\boldsymbol{\ell}_\text{p}\approx\boldsymbol{\ell}$ is available, where $\boldsymbol{\ell}_\text{p}\in\mathcal{L}$, we aim to select set $\mathcal{M}$ such that the CRLB obtained on $\boldsymbol{\ell}_\text{p}$ is minimized.
We formulate the dynamic sensor selection problem as
\begin{align}
    (\boldsymbol{\mathcal{P}}_\text{D}):~\mathcal{M}_\text{D}^\star=&\argmin_{\mathcal{M}}\;\sigma^2_{\mathcal{M}}(\boldsymbol{\ell}_\text{p}) \label{eq:dynamic_min_prob} \\
    \text{s.t.}\;\;&\vert \mathcal{M}\vert  = M,\;\mathcal{M}\subseteq\mathcal{M}_\text{max}. \label{eq:dynamic_min_const1}
\end{align}

\subsubsection{Robust sensor selection} When prior information on $\boldsymbol{\ell}_\text{p} \in \mathcal{L}$ is not available, we aim to minimize the worst-case positioning error across the potential target locations in $\mathcal{L}$.
We subsequently formulate the robust sensor selection problem as
\begin{align}
    (\boldsymbol{\mathcal{P}}_\text{R}):~\mathcal{M}_\text{R}^\star=&\argmin_{\mathcal{M}}\max_{\boldsymbol{\ell}^\mathsf{t}_g\in \mathcal{L}}\;\sigma^2_{\mathcal{M}}(\boldsymbol{\ell}^\mathsf{t}_g) \label{eq:robust_minmax_prob} \\
    \text{s.t.}\;\;&\vert \mathcal{M}\vert  = M,\;\mathcal{M}\subseteq\mathcal{M}_\text{max}. \label{eq:robust_minmax_const1}
\end{align}
The definition of our robust sensor selection problem in $\boldsymbol{\mathcal{P}}_\text{R}$ resembles the existing formalization of robustness in selection-based optimization problems~\cite{Krause08,Powers16,Bajovic11}.

Both sensor selection problems are combinatorial optimizations and can be solved in theory via exhaustive search over their feasible spaces.
In practice, however, this presents scalability challenges.
Particularly, the solution to $\boldsymbol{\mathcal{P}}_\text{D}$ can be found with complexity of $\mathcal{O}\big(\frac{M_\text{max}!}{(M_\text{max}-M)!M!}\big)$~\cite{Zhao19}, which can become prohibitive as $M_{\text{max}}$ increases.
For~$\boldsymbol{\mathcal{P}}_\text{R}$, the complexity is additionally impacted by the size of $\mathcal{L}$ because what we are trying to minimize is the \emph{max} function over $\boldsymbol{\ell}^\mathsf{t}_g$.
This motivates us to develop more computationally efficient approaches, beginning with analysis of the CRLB in Section~\ref{sec:CRLB}.

\vspace{-0.5mm}
\section{CRLB for Wireless Positioning}\label{sec:CRLB}

In this section, we obtain the CRLB expression in two different forms. Unique properties observed in each form motivate our development of sensor selection strategies that improve accuracy and/or complexity.
It is known that the CRLB can be computed as $\sigma^2_{\mathcal{M}}(\boldsymbol{\ell})=\textrm{tr}\{\mathcal{I}^{-1}_{\mathcal{M}}(\boldsymbol{\ell})\}$, where $\mathcal{I}_{\mathcal{M}}(\boldsymbol{\ell})$ is the Fisher information matrix (FIM)~\cite{Kay97} for 3D wireless positioning with hybrid TOA/RSS measurements on the target location $\boldsymbol{\ell}$ using the sensors in $\mathcal{M}$.
Hence, we first derive the generalized expression of $\mathcal{I}_{\mathcal{M}}(\boldsymbol{\ell})$ for our setting and then simplify the expression to obtain two closed-form CRLB expressions.

The FIM for our problem setup is given by~\cite{Kay97}
\begin{align}
    \hspace{-1mm}\mathcal{I}_{\mathcal{M}}(\boldsymbol{\ell}) = -\mathbb{E}\hspace{-1mm}
    \begin{bmatrix}
        \frac{\partial^2 l_{\mathcal{M}}(\widehat{\boldsymbol{q}}_{\mathcal{M}}\vert \boldsymbol{\ell})}{\partial x^2} \hspace{-2mm}& \frac{\partial^2 l_{\mathcal{M}}(\widehat{\boldsymbol{q}}_{\mathcal{M}}\vert \boldsymbol{\ell})}{\partial x \partial y} \hspace{-2mm}& \frac{\partial^2 l_{\mathcal{M}}(\widehat{\boldsymbol{q}}_{\mathcal{M}}\vert \boldsymbol{\ell})}{\partial x \partial z} \\
        \frac{\partial^2 l_{\mathcal{M}}(\widehat{\boldsymbol{q}}_{\mathcal{M}}\vert \boldsymbol{\ell})}{\partial  y \partial x} \hspace{-2mm}& \frac{\partial^2 l_{\mathcal{M}}(\widehat{\boldsymbol{q}}_{\mathcal{M}}\vert \boldsymbol{\ell})}{\partial y^2} \hspace{-2mm}& \frac{\partial^2 l_{\mathcal{M}}(\widehat{\boldsymbol{q}}_{\mathcal{M}}\vert \boldsymbol{\ell})}{\partial y \partial z} \\
        \frac{\partial^2 l_{\mathcal{M}}(\widehat{\boldsymbol{q}}_{\mathcal{M}}\vert \boldsymbol{\ell})}{\partial z \partial x} \hspace{-2mm}& \frac{\partial^2 l_{\mathcal{M}}(\widehat{\boldsymbol{q}}_{\mathcal{M}}\vert \boldsymbol{\ell})}{\partial z \partial y} \hspace{-2mm}& \frac{\partial^2 l_{\mathcal{M}}(\widehat{\boldsymbol{q}}_{\mathcal{M}}\vert \boldsymbol{\ell})}{\partial z^2}
    \end{bmatrix}\hspace{-1mm},
    \label{eq:FIM_hybrid_1}
\end{align}
where
\begin{align}
    l_{\mathcal{M}}(\widehat{\boldsymbol{q}}_{\mathcal{M}}\vert \boldsymbol{\ell})=&-(1/2)(\widehat{\boldsymbol{q}}_{\mathcal{M}}-\boldsymbol{q}_{\mathcal{M}})^\top\mathbf{R}_{\mathcal{M}}^{-1}(\widehat{\boldsymbol{q}}_{\mathcal{M}}-\boldsymbol{q}_{\mathcal{M}}) \nonumber \\
    &\;\:-\ln\hspace{-0.5mm}\Big((2\pi)^{M}\Big(\hspace{-1mm}\prod_{m\in\mathcal{M}}\hspace{-1mm}\widehat{d}_{\text{R},m}\Big)\det(\mathbf{R}_{\mathcal{M}})^{\frac{1}{2}}\Big)
    \label{eq:log-likelihood_hybrid}
\end{align}
is the log-likelihood function derived from~\eqref{eq:joint_hybrid_PDF}.
Based on~\eqref{eq:FIM_hybrid_1} and~\eqref{eq:log-likelihood_hybrid}, we write
\begin{align}
    \mathcal{I}_{\mathcal{M}}(\boldsymbol{\ell}) & =
    \begin{bmatrix}
        \mathcal{I}^{(xx)}_{\mathcal{M}} & \mathcal{I}^{(xy)}_{\mathcal{M}} & \mathcal{I}^{(xz)}_{\mathcal{M}} \\
        \mathcal{I}^{(yx)}_{\mathcal{M}} & \mathcal{I}^{(yy)}_{\mathcal{M}} & \mathcal{I}^{(yz)}_{\mathcal{M}} \\
        \mathcal{I}^{(zx)}_{\mathcal{M}} & \mathcal{I}^{(zy)}_{\mathcal{M}} & \mathcal{I}^{(zz)}_{\mathcal{M}}
    \end{bmatrix},
    \label{eq:FIM_hybrid_2}
\end{align}
where the elements can be expressed as
\begin{equation}
    \mathcal{I}_{\mathcal{M}}^{(vw)} =\frac{\partial\boldsymbol{q}_{\mathcal{M}}^\top}{\partial v}\mathbf{R}_{\mathcal{M}}^{-1}\frac{\partial\boldsymbol{q}_{\mathcal{M}}}{\partial w}, \quad \forall v,w \in \{x,y,z\}.
    \label{eq:FIM_hybrid_entry}
\end{equation}
The derivation of~\eqref{eq:FIM_hybrid_entry} is provided in Appendix~\ref{appendix:FIM}.

We assume zero correlation on the noises among different sensors~\cite{Laaraiedh12,Li18,Shen12}, i.e., $\rho_{m_im_j}=1$ and $\eta_{m_im_j}\in[0,1)$ if $i\hspace{-1mm}=\hspace{-1mm}j$, and $\rho_{m_im_j}\hspace{-1mm}=0$ and $\eta_{m_im_j}\hspace{-1mm}=0$ if $i\hspace{-1mm}\neq\hspace{-1mm}j$.
This makes $\mathbf{R}_{\mathcal{M}}$ a block diagonal matrix, i.e., $\mathbf{R}_{\mathcal{M}}=\text{diag}(\mathbf{R}_{m_1m_1},$ $\mathbf{R}_{m_2m_2},\ldots,\mathbf{R}_{m_Mm_M})$.
The inverse of $\mathbf{R}_{\mathcal{M}}$ then becomes
\begin{equation}
    \mathbf{R}_{\mathcal{M}}^{-1}=\text{diag}(\mathbf{R}_{m_1m_1}^{-1},\mathbf{R}_{m_2m_2}^{-1},\ldots,\mathbf{R}_{m_Mm_M}^{-1}).
    \label{eq:block_matrix_inverse}
\end{equation}
Using~\eqref{eq:block_matrix_inverse} and $d_{m}\hspace{-0.5mm}=[(x^\mathsf{s}_{m}\hspace{-0.5mm}-x)^2+(y^\mathsf{s}_{m}\hspace{-0.5mm}-y)^2+(z^\mathsf{s}_{m}\hspace{-0.5mm}-z)^2]^\frac{1}{2}$, \eqref{eq:FIM_hybrid_entry} can be rewritten as
\begin{align}
    &\mathcal{I}_{\mathcal{M}}^{(vw)} =\sum_{m\in\mathcal{M}}\frac{\partial\boldsymbol{q}_{m}^\top}{\partial v}\mathbf{R}_{mm}^{-1}\frac{\partial\boldsymbol{q}_{m}}{\partial w} \nonumber \\
    &\hspace{-2mm}=\hspace{-2mm}\sum_{m\in\mathcal{M}}\hspace{-1mm}\left[\hspace{-0.5mm}{-\frac{(v^\mathsf{s}_{m}\hspace{-1mm}-v)}{d_{m}} \atop -\frac{(v^\mathsf{s}_{m}\hspace{-1mm}-v)}{d_{m}^2}}\hspace{-1mm}\right]^{\hspace{-1.2mm}\top}\hspace{-2mm}
    \begin{bmatrix}
        \hspace{-0.5mm}\frac{\sigma^{-2}_{\text{T},m}}{(1-\eta^2_{mm})} \hspace{-2mm}&\hspace{-2mm}\frac{-\eta_{mm}\sigma^{-1}_{\text{T},m}}{(1-\eta^2_{mm})\sigma_{\text{R},m}}\hspace{-0.5mm} \nonumber \\
        \hspace{-0.5mm}\frac{-\eta_{mm}\sigma^{-1}_{\text{T},m}}{(1-\eta^2_{mm})\sigma_{\text{R},m}} \hspace{-2mm}&\hspace{-2mm} \frac{\sigma^{-2}_{\text{R},m}}{(1-\eta^2_{mm})}\hspace{-0.5mm}
    \end{bmatrix}
    \hspace{-2mm}\left[\hspace{-0.5mm}{-\frac{(w^\mathsf{s}_{m}\hspace{-1mm}-w)}{d_{m}} \atop -\frac{(w^\mathsf{s}_{m}\hspace{-1mm}-w)}{d_{m}^2}}\hspace{-1mm}\right]\hspace{-2mm} \\
    &\hspace{-2mm}=\hspace{-2mm}\sum_{m\in\mathcal{M}}\epsilon_{m}\frac{(v^\mathsf{s}_m-v)(w^\mathsf{s}_m-w)}{d^2_{m}},
    \label{eq:FIM_hybrid_entry_2}
\end{align}
where
\begin{equation}
    \epsilon_{m} = \frac{\sigma^{-2}_{\text{T},m}}{(1-\eta_{mm}^2)}+\frac{\sigma^{-2}_{\text{R},m}}{(1-\eta_{mm}^2)d^2_{m}}\hspace{-0.5mm}-\frac{2\eta_{mm}\sigma^{-1}_{\text{T},m}\sigma^{-1}_{\text{R},m}}{(1-\eta_{mm}^2)d_{m}}.
    \label{eq:corr_epsilon}
\end{equation}
Using~\eqref{eq:FIM_hybrid_entry_2}, \eqref{eq:FIM_hybrid_2} can be expressed as
\begin{align}
    \hspace{-1mm} \mathcal{I}_{\mathcal{M}}(\boldsymbol{\ell})\hspace{-.5mm} &= \hspace{-2mm} \sum_{m\in\mathcal{M}}\hspace{-1mm}\hspace{-.5mm}\epsilon_m\hspace{-1mm}\hspace{-.5mm}
    \begin{bmatrix}
        \hspace{-.5mm}
        \frac{(x^\mathsf{s}_m-x)(x^\mathsf{s}_m-x)}{d^2_{m}} \hspace{-1.5mm}&\hspace{-1.5mm} \frac{(x^\mathsf{s}_m-x)(y^\mathsf{s}_m-y)}{d^2_{m}} \hspace{-1.5mm}&\hspace{-1.5mm} \frac{(x^\mathsf{s}_m-x)(z^\mathsf{s}_m-z)}{d^2_{m}} \\
        \hspace{-.5mm}
        \frac{(y^\mathsf{s}_m-y)(x^\mathsf{s}_m-x)}{d^2_{m}} \hspace{-1.5mm}&\hspace{-1.5mm} \frac{(y^\mathsf{s}_m-y)(y^\mathsf{s}_m-y)}{d^2_{m}} \hspace{-1.5mm}&\hspace{-1.5mm} \frac{(y^\mathsf{s}_m-y)(z^\mathsf{s}_m-z)}{d^2_{m}} \\
        \hspace{-.5mm}
        \frac{(z^\mathsf{s}_m-z)(x^\mathsf{s}_m-x)}{d^2_{m}} \hspace{-1.5mm}&\hspace{-1.5mm} \frac{(z^\mathsf{s}_m-z)(y^\mathsf{s}_m-y)}{d^2_{m}} \hspace{-1.5mm}&\hspace{-1.5mm} \frac{(z^\mathsf{s}_m-z)(z^\mathsf{s}_m-z)}{d^2_{m}}
        \hspace{-.5mm}
    \end{bmatrix}
    \hspace{-2mm} \nonumber \\
    & = \hspace{-2mm} \sum_{m\in\mathcal{M}}\epsilon_{m}\boldsymbol{u}_{m}\boldsymbol{u}_{m}^\top, \label{eq:FIM_hybrid_4}
\end{align}
where $\boldsymbol{u}_{m} = [\nicefrac{(x^\mathsf{s}_m-x)}{d_{m}}, \nicefrac{(y^\mathsf{s}_m-y)}{d_{m}}, \nicefrac{(z^\mathsf{s}_m-z)}{d_{m}}]^\top$ is the normalized LOS vector between sensor $m$ and the target.
Using~\eqref{eq:FIM_hybrid_4} for expressing $\mathcal{I}_{\mathcal{M}}(\boldsymbol{\ell})$, the \textbf{trace form of the CRLB (T-CRLB)} is obtained as
\begin{equation}
    \sigma^2_{\mathcal{M}}(\boldsymbol{\ell})=\textrm{tr}\left\{\bigg(\sum_{m\in\mathcal{M}}\epsilon_{m}\boldsymbol{u}_{m}\boldsymbol{u}_{m}^\top\bigg)^{-1}\right\}.
    \label{eq:CRLB_hybrid_corr_1}
\end{equation}

The T-CRLB reveals that the CRLB is a function of $M$ rank-one positive semidefinite matrices, each of which corresponds to one of the sensors in $\mathcal{M}$.
With the CRLB taking this form, optimizing the bound can be perceived as a standard \emph{E}-optimality experiment design problem~\cite{Boyd04}, which is known to be convex.

Without loss of generality, we can replace $\boldsymbol{\ell}$ in~\eqref{eq:CRLB_hybrid_corr_1} with $\boldsymbol{\ell}^\mathsf{t}_g$ to represent the CRLB for a discretized potential target location $\boldsymbol{\ell}^\mathsf{t}_g\in\mathcal{L}$.
The expression of $\sigma^2_{\mathcal{M}}(\boldsymbol{\ell}^\mathsf{t}_g)$ is obtained by replacing all instances of $d_m$ in~\eqref{eq:CRLB_hybrid_corr_1} with $d_{mg}=\|\boldsymbol{\ell}^\mathsf{s}_m-\boldsymbol{\ell}^\mathsf{t}_g\|_2$.
The resulting expression is given by
\begin{equation}
    \sigma^2_{\mathcal{M}}(\boldsymbol{\ell}^\mathsf{t}_g)=\textrm{tr}\left\{\bigg(\sum_{m\in\mathcal{M}}\epsilon_{mg}\boldsymbol{u}_{mg}\boldsymbol{u}_{mg}^\top\bigg)^{-1}\right\},
    \label{eq:CRLB_hybrid_corr_2}
\end{equation}
where $\boldsymbol{u}_{mg}=\big[\nicefrac{(x^\mathsf{s}_m-x^\mathsf{t}_g)}{d_{mg}}, \nicefrac{(y^\mathsf{s}_m-y^\mathsf{t}_g)}{d_{mg}}, \nicefrac{(z^\mathsf{s}_m-z^\mathsf{t}_g)}{d_{mg}}\big]^\top$ and $\epsilon_{mg} = \frac{\sigma^{-2}_{\text{T},m}}{(1-\eta_{mm}^2)}\hspace{-0.2mm}+\hspace{-0.2mm}\frac{\sigma^{-2}_{\text{R},m}}{(1-\eta_{mm}^2)d^2_{mg}}\hspace{-0.2mm}-\hspace{-0.2mm}\frac{2\eta_{mm}\sigma^{-1}_{\text{T},m}\sigma^{-1}_{\text{R},m}}{(1-\eta_{mm}^2)d_{mg}}$.
In Section~\ref{sec:robust}, we will utilize~\eqref{eq:CRLB_hybrid_corr_2} to formulate different versions of the robust sensor selection problem that are equivalent to~$\boldsymbol{\mathcal{P}}_\text{R}$.

The T-CRLB allows us to evaluate the CRLB using $\boldsymbol{u}_m$ and $\epsilon_m$. However, the inverse operation prevents us from directly observing the relationship between the selected sensors and positioning accuracy. Therefore, from~\eqref{eq:CRLB_hybrid_corr_1}, we continue our derivation and obtain the \textbf{fractional form of CRLB (F-CRLB)} based on the following proposition proven in Appendix~\ref{appendix:CRLB}.
\begin{proposition}
    Given $2M$ distance estimates $\widehat{\boldsymbol{q}}_{\mathcal{M}}$ obtained from the hybrid TOA/RSS measurements acquired by $M$ sensors in $\mathcal{M}$ following the PDF in~\eqref{eq:joint_hybrid_PDF}, the F-CRLB is given by~\eqref{eq:CRLB_derived}.
    \label{prop:CRLB}
    \vspace{-1mm}
\end{proposition}

\begin{table*}[t]
\vspace{-2mm}
\begin{minipage}{0.99\textwidth}
\begin{equation}
    \sigma^2_{\mathcal{M}}(\boldsymbol{\ell}) = \frac{N_{\mathcal{M}}(\boldsymbol{\ell})}{D_{\mathcal{M}}(\boldsymbol{\ell})} = \frac{\sum_{m_1\in\mathcal{M}} \sum_{\substack{m_2\in\mathcal{M}\\m_2>m_1}} \epsilon_{m_1}\epsilon_{m_2}\sin^2\theta_{m_1m_2}} {\sum_{m_1\in\mathcal{M}} \sum_{\substack{m_2\in\mathcal{M}\\m_2>m_1}} \sum_{\substack{m_3\in\mathcal{M}\\m_3>m_2}} \epsilon_{m_1}\epsilon_{m_2}\epsilon_{m_3}\sin^2{\theta_{m_1m_2}}\sin^2{\phi_{m_1m_2m_3}}}
    \label{eq:CRLB_derived}
\end{equation}
\hrule
\vspace{-2mm}
\end{minipage}
\end{table*}

In~\eqref{eq:CRLB_derived}, $\theta_{m_1m_2}$ is the angle between the LOS vectors $\boldsymbol{u}_{m_1}$ and $\boldsymbol{u}_{m_2}$, and $\phi_{m_1m_2m_3}$ is the angle between the vector $\boldsymbol{u}_{m_3}$ and the plane containing $\boldsymbol{u}_{m_1}$ and $\boldsymbol{u}_{m_2}$.
A visualization of these angle parameters for $\mathcal{M}=\{1,2,3\}$ is provided in Fig.~\ref{sensors_layout}.
Compared to the T-CRLB, the F-CRLB offers interpretations on the relationship between sensor placement and the resulting CRLB.
For example, the bound becomes undefined whenever $D_{\mathcal{M}}(\boldsymbol{\ell})$ yields zero, and this singularity occurs when the selected sensors have a co-planar arrangement (i.e., the 3D coordinates of $\mathcal{M}$ can be contained by a single plane).
In Section~\ref{sec:dynamic}, where we focus on our dynamic sensor selection problem~$\boldsymbol{\mathcal{P}}_\text{D}$, one of our sensor selection strategies will be based on the unique characteristics found in the F-CRLB.

\vspace{-1mm}
\section{Dynamic Sensor Selection Strategies}\label{sec:dynamic}

In this section, we focus on $\boldsymbol{\mathcal{P}}_\text{D}$, where the sensor set $\mathcal{M}$ is selected such that $\sigma^2_\mathcal{M}(\boldsymbol{\ell})$, for a given $\boldsymbol{\ell}$, is minimized. 
Despite that SDR provides near-optimal performance to solve the problem~\cite{Zhao19,Dai20}, its complexity becomes prohibitive for a large number of sensors (e.g., $\mathcal{O}(M_\text{max}^{4.5})$~\cite{Dai20} from using interior point methods~\cite{Boyd04}).
To overcome this issue, greedy selection strategies~\cite{Godrich12,Zhao19} have been proposed for improved complexity.
For further reducing complexity, we propose two selection metrics that do not fully compute the CRLB yet are still effective for greedy selection.
The metrics are designed from the properties found in each of two forms: T-CRLB and F-CRLB.
We develop greedy selection algorithms based on each proposed metric.
Then, we present our complexity analysis in Section~\ref{ssec:dynamic_complexity}.
Note that our proposed sensor selection algorithms can also be applied when only TOA or RSS is available for the measurement since the resulting CRLB can be derived to take the same form as T-CRLB or F-CRLB~\cite{Patwari03}.

\vspace{-3mm}
\subsection{Sensor Selection based on the T-CRLB}\label{ssec:dynamic_TCRLB}

Solving $\boldsymbol{\mathcal{P}}_\text{D}$ via greedy selection involves iterative selection steps, in each of which the single most promising sensor for minimizing the CRLB is added to the set.
To optimize this process, we first quantify the marginal CRLB reduction from selecting a sensor.
Defining $\mathcal{M}_i$ to be the set of sensors selected by the first $i$ steps in greedy selection, i.e., $i \in \{1,\ldots,M\}$ and $\vert \mathcal{M}_i\vert =i$, we can write the marginal CRLB reduction achieved by the selection step $i$ as $\sigma^2_{\mathcal{M}_{i-1}}(\boldsymbol{\ell})-\sigma^2_{\mathcal{M}_{i}}(\boldsymbol{\ell})$.
Using the T-CRLB in~\eqref{eq:CRLB_hybrid_corr_1}, we present a simplified expression for the marginally reduced CRLB in the following proposition.

\begin{figure}[t]
    \centering
    \includegraphics[width=0.45\linewidth]{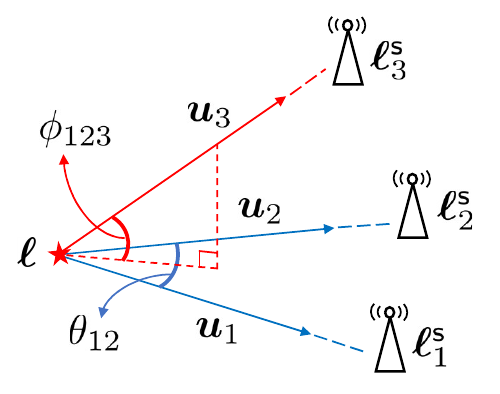}
    \caption{A visual illustration of $\theta_{m_1m_2}$ and $\phi_{m_1m_2m_3}$ in the F-CRLB expression for a given target and a sensor set $\mathcal{M} = \{1,2,3\}$.}
    \label{sensors_layout}
    \vspace{-2mm}
\end{figure}

\begin{proposition}
    For the greedy sensor selection step $i$, the marginal CRLB reduction achieved by selecting sensor $m\in\mathcal{M}_\text{max}\backslash \mathcal{M}_{i-1}$ is expressed as
    \begin{equation}
        \hspace{-0mm}\sigma^2_{\mathcal{M}_{i-1}}\hspace{-0.5mm}(\boldsymbol{\ell})\hspace{-0.5mm}-\hspace{-0.5mm}\sigma^2_{\mathcal{M}_{i}}\hspace{-0.5mm}(\boldsymbol{\ell})\hspace{-0.5mm}=\hspace{-0.5mm}\textrm{tr}\left\{\hspace{-0.5mm}\frac{\epsilon_{m}\mathcal{I}_{\mathcal{M}_{i-1}}^{-1}\hspace{-0.5mm}(\boldsymbol{\ell})\boldsymbol{u}_{m}\boldsymbol{u}_{m}^\top\mathcal{I}_{\mathcal{M}_{i-1}}^{-1}\hspace{-0.5mm}(\boldsymbol{\ell})}{1+\epsilon_{m}\boldsymbol{u}_{m}^\top\mathcal{I}_{\mathcal{M}_{i-1}}^{-1}\hspace{-0.5mm}(\boldsymbol{\ell})\boldsymbol{u}_{m}}\hspace{-0.5mm}\right\}\hspace{-0.5mm}.\hspace{-7mm}\label{eq:marginal_CRLB}
    \end{equation}
    \label{prop:marginal_CRLB}
\end{proposition}
\vspace{-4mm}
\begin{proof}
    Via~\eqref{eq:FIM_hybrid_4}, the FIM of $\mathcal{M}_{i-1}$ is expressed as $\mathcal{I}_{\mathcal{M}_{i-1}}\hspace{-0.5mm}(\boldsymbol{\ell})\hspace{-0.5mm}=\hspace{-0.5mm}\sum_{m\in\mathcal{M}_{i-1}}\epsilon_m\boldsymbol{u}_m\boldsymbol{u}_m^\top$, which is the sum of $i-1$ rank-one matrices.
    Then, adding an additional sensor to $\mathcal{M}_{i-1}$ is equivalent to applying a rank-one matrix update to $\mathcal{I}_{\mathcal{M}_{i-1}}$.
    Since $\mathcal{M}_{i}=\mathcal{M}_{i-1}\cup \{m\}$, we can write $\mathcal{I}_{\mathcal{M}_i}(\boldsymbol{\ell})=\mathcal{I}_{\mathcal{M}_{i-1}}(\boldsymbol{\ell})+\epsilon_m\boldsymbol{u}_m\boldsymbol{u}_m^\top$.
    Using the relationship $\sigma^2_{\mathcal{M}}(\boldsymbol{\ell})=\textrm{tr}\{\mathcal{I}^{-1}_{\mathcal{M}}(\boldsymbol{\ell})\}$,
    \begin{align}
        \sigma^2_{\mathcal{M}_{i-1}}&(\boldsymbol{\ell})-\sigma^2_{\mathcal{M}_{i}}(\boldsymbol{\ell}) =\textrm{tr}\{\mathcal{I}_{\mathcal{M}_{i-1}}^{-1}(\boldsymbol{\ell})\}-\textrm{tr}\{\mathcal{I}^{-1}_{\mathcal{M}_i}(\boldsymbol{\ell})\} \nonumber \\
        &=\textrm{tr}\{\mathcal{I}_{\mathcal{M}_{i-1}}^{-1}(\boldsymbol{\ell}) -(\mathcal{I}_{\mathcal{M}_{i-1}}(\boldsymbol{\ell})+\epsilon_m\boldsymbol{u}_m\boldsymbol{u}_m^\top)^{-1}\} \nonumber \\
        &=\textrm{tr}\left\{\frac{\epsilon_{m}\mathcal{I}_{\mathcal{M}_{i-1}}^{-1}\boldsymbol{u}_{m}\boldsymbol{u}_{m}^\top\mathcal{I}_{\mathcal{M}_{i-1}}^{-1}}{1+\epsilon_{m}\boldsymbol{u}_{m}^\top\mathcal{I}_{\mathcal{M}_{i-1}}^{-1}\boldsymbol{u}_{m}}\right\}, \label{eq:marginal_CRLB_proof_3}
    \end{align}
    where~\eqref{eq:marginal_CRLB_proof_3} is from the Sherman-Morrison formula~\cite{Sherman50}.
\end{proof}

We see from Proposition~\ref{prop:marginal_CRLB} that the marginal CRLB reduction depends on (i) the FIM before making the selection and (ii) the rank-one matrix corresponding to sensor $m$.
It should be noted that the inverse operation is not required to compute $\mathcal{I}_{\mathcal{M}_{i}}^{-1}$ because it can be obtained from the Sherman-Morrison formula and the previous iteration.
One issue with this metric is that $\mathcal{I}_{\mathcal{M}_{i-1}}$ must be always non-singular, i.e., the matrix must be invertible.
To ensure this, at least three sensors must be included in $\mathcal{M}_{i-1}$ for each $i$, which makes the greedy selection based on~\eqref{eq:marginal_CRLB} applicable only for $i \geq 4$.
The same issue exists in the greedy selection algorithm~\cite{Zhao19} where the complete CRLB expression is used as a selection metric.
As a result, we select the first three sensors, comprising $\mathcal{M}_3$, heuristically (e.g., random selection) and conduct the rest of greedy selection using the proposed metric.

Formally, for each selection step $i$, where $4\leq i \leq M$, we solve the optimization problem
\begin{equation}
    m_i^\star=\argmax_{m\in\mathcal{M}_\text{max}\backslash \mathcal{M}_{i-1}} \textrm{tr}\left\{\frac{\epsilon_{m}\mathcal{I}_{\mathcal{M}_{i-1}}^{-1}\boldsymbol{u}_{m}\boldsymbol{u}_{m}^\top\mathcal{I}_{\mathcal{M}_{i-1}}^{-1}}{1+\epsilon_{m}\boldsymbol{u}_{m}^\top\mathcal{I}_{\mathcal{M}_{i-1}}^{-1}\boldsymbol{u}_{m}}\right\}
    \label{eq:dynamic_trace}
\end{equation}
via exhaustive search and update $\mathcal{M}_i=\mathcal{M}_{i-1}\cup \{m_i^\star\}$.
Once all $M$ sensors are selected, $\mathcal{M}_M$ is declared as a solution.
The overall procedure is summarized in Algorithm~\ref{alg:dynamic_trace}.
\begin{algorithm}[ht]
{\small
\caption{Greedy Sensor Selection based on~\eqref{eq:dynamic_trace}}
\label{alg:dynamic_trace}
\begin{algorithmic}
\REQUIRE $\epsilon_m$ and $\boldsymbol{u}_m$, $\forall m \in \mathcal{M}_\text{max}$
\STATE Generate $\mathcal{M}_3:$ 3 randomly selected sensors from $\mathcal{M}_\text{max}$
\STATE $\mathcal{I}_{\mathcal{M}_{3}}^{-1} = \big(\sum_{m\in\mathcal{M}_3}\epsilon_m\boldsymbol{u}_m\boldsymbol{u}_m^\top\big)^{-1}$
\STATE $i=4$
\WHILE {$i \leq M$}
    \STATE Find $m_i^\star$ from solving~\eqref{eq:dynamic_trace} via exhaustive search
    \STATE $\mathcal{M}_i = \mathcal{M}_{i-1} \cup \{m_i^\star\}$
    \STATE $\mathcal{I}_{\mathcal{M}_{i}}^{-1} = \mathcal{I}_{\mathcal{M}_{i-1}}^{-1} - \frac{\epsilon_{m_i^\star}\mathcal{I}_{\mathcal{M}_{i-1}}^{-1}\boldsymbol{u}_{m_i^\star}\boldsymbol{u}_{m_i^\star}^\top\mathcal{I}_{\mathcal{M}_{i-1}}^{-1}}{1+\epsilon_{m_i^\star}\boldsymbol{u}_{m_i^\star}^\top\mathcal{I}_{\mathcal{M}_{i-1}}^{-1}\boldsymbol{u}_{m_i^\star}}$
    \STATE $i = i + 1$
\ENDWHILE
\RETURN $\mathcal{M}_M$ 
\end{algorithmic}
}
\end{algorithm}

\vspace{-5mm}
\subsection{Sensor Selection based on the F-CRLB}\label{ssec:dynamic_FCRLB}

According to the F-CRLB in~\eqref{eq:CRLB_derived}, we can compute $\sigma^2_\mathcal{M}(\boldsymbol{\ell})$ by separately evaluating $N_{\mathcal{M}}(\boldsymbol{\ell})$ and $D_{\mathcal{M}}(\boldsymbol{\ell})$, which are the sums of ${M}\choose{2}$ pairs and ${M}\choose{3}$ triplets, respectively, generated out of $\mathcal{M}$.
We can exploit this pattern for our greedy sensor selection strategy.
Suppose an algorithm is in selection step $i$ and attempts to select a single sensor from the remaining set $\mathcal{M}_\text{max}\backslash\mathcal{M}_{i-1}$ using the F-CRLB as its metric.
If we define $A_{ab}=\epsilon_{a}\epsilon_{b}\sin^2\theta_{ab}$ and $V_{abc}=\epsilon_{a}\epsilon_{b}\epsilon_{c}\sin^2{\theta_{ab}}\sin^2{\phi_{abc}}$ for $a,b,c\in\mathcal{M}_\text{max}$, the optimization problem can be written as
\begin{equation}
    m_i^\star=\argmin_{m\in\mathcal{M}_\text{max}\backslash \mathcal{M}_{i-1}}
    \frac{N_{\mathcal{M}_{i-1}}(\boldsymbol{\ell})+A^\mathsf{sum}_{\mathcal{M}_{i-1},m}(\boldsymbol{\ell})} {D_{\mathcal{M}_{i-1}}(\boldsymbol{\ell})+V^\mathsf{sum}_{\mathcal{M}_{i-1},m}(\boldsymbol{\ell})},
    \label{eq:dynamic_fraction}
\end{equation}
where $A^\mathsf{sum}_{\mathcal{M}_{i-1},m}(\boldsymbol{\ell})=\sum_{m_1\in\mathcal{M}_{i-1}}A_{m_1m}$ and $V^\mathsf{sum}_{\mathcal{M}_{i-1},m}(\boldsymbol{\ell})=\sum_{m_1\in\mathcal{M}_{i-1}}\hspace{-1mm}\sum_{\substack{m_2\in\mathcal{M}_{i-1}\\m_2>m_1}}V_{m_1m_2m}$.
Since both $N_{\mathcal{M}_{i-1}}(\boldsymbol{\ell})$ and $D_{\mathcal{M}_{i-1}}(\boldsymbol{\ell})$ are available from the previous selection step $i-1$, instead of computing ${{i}\choose{2}}+{{i}\choose{3}}$ summation terms, only $(i-1)+{{i-1}\choose{2}}$ terms are required to evaluate the objective function of~\eqref{eq:dynamic_fraction} for each value of $m$.

To further reduce the complexity of this sensor selection strategy, we introduce the following proposition.
\begin{proposition}
    Let $\lambda_{\mathcal{M},1}$, $\lambda_{\mathcal{M},2}$, and $\lambda_{\mathcal{M},3}$ denote the three eigenvalues of $\mathcal{I}_\mathcal{M}(\boldsymbol{\ell})$. The CRLB $\sigma^2_\mathcal{M}(\boldsymbol{\ell})$ is half the surface to volume ratio (SVR) of a rectangular prism with dimension $\lambda_{\mathcal{M},1} \times \lambda_{\mathcal{M},2} \times \lambda_{\mathcal{M},3}$.
    \label{prop:SVR}
\end{proposition}
The proof for Proposition~\ref{prop:SVR} is given in Appendix~\ref{appendix:SVR}.
The key takeaway is how the CRLB can be characterized by the geometry of the eigenvalues of $\mathcal{I}_\mathcal{M}(\boldsymbol{\ell})$.
Particularly, the F-CRLB in~\eqref{eq:CRLB_derived} can be paired with the SVR such that $N_{\mathcal{M}}(\boldsymbol{\ell})$ and $D_{\mathcal{M}}(\boldsymbol{\ell})$ represent the surface area and volume, respectively.
In the following, we present a different method for minimizing the CRLB based on this perspective.

According to Proposition~\ref{prop:SVR}, the solution to problem $\boldsymbol{\mathcal{P}}_\text{D}$ (i.e., $\mathcal{M}_\text{D}^\star$) should minimize the SVR of the rectangular prism defined by the eigenvalues of $\mathcal{I}_\mathcal{M}(\boldsymbol{\ell})$.
In other words, our dynamic sensor selection problem is equivalent to finding the set of sensors that minimizes the SVR of the resulting rectangular prism.
To decrease the SVR, we desire a rectangular prism with (i) larger size (i.e., greater eigenvalues) and (ii) more cubical shape (i.e., a smaller condition number).
Note that the range in which our eigenvalues can vary is fundamentally limited since we only consider $M_\text{max}$ sensors placed within the confined space.
Therefore, we can focus on the first condition (i.e., the size) to minimize the SVR. Note that a rectangular prism with larger volume tends to yield a lower SVR: a higher volume also implies a larger surface area, but volume has a higher rate of change for a unit increase in dimension.

This discussion indicates that we can conduct our sensor selection by relying on either $N_{\mathcal{M}}(\boldsymbol{\ell})$ or $D_{\mathcal{M}}(\boldsymbol{\ell})$.
We thus propose $A^\mathsf{sum}_{\mathcal{M}_{i-1},m}(\boldsymbol{\ell})$ and $V^\mathsf{sum}_{\mathcal{M}_{i-1},m}(\boldsymbol{\ell})$ as metrics that will make our greedy sensor selection more computationally efficient.
For each greedy selection step $i$, the sensor to be selected is determined by solving either of the following optimization problems:
\begin{equation}
    m_i^\star=\argmax_{m\in\mathcal{M}_\text{max}\backslash \mathcal{M}_{i-1}} A^\mathsf{sum}_{\mathcal{M}_{i-1},m}(\boldsymbol{\ell})
    \label{eq:dynamic_numerator}
\end{equation}
or
\begin{equation}
    m_i^\star=\argmax_{m\in\mathcal{M}_\text{max}\backslash \mathcal{M}_{i-1}} V^\mathsf{sum}_{\mathcal{M}_{i-1},m}(\boldsymbol{\ell}).
    \label{eq:dynamic_denominator}
\end{equation}

Using~\eqref{eq:dynamic_numerator} for the selection criterion has lower computational complexity than~\eqref{eq:dynamic_denominator}.
Also, with~\eqref{eq:dynamic_numerator}, one can start conducting the greedy selection as early as $i=2$, whereas with~\eqref{eq:dynamic_denominator}, sensors must be selected heuristically until $i = 3$ due to the way in which the summations are formed.
However, solely relying on~\eqref{eq:dynamic_numerator} can result in a low SVR when selected sensors are in co-planar arrangement.
Using~\eqref{eq:dynamic_denominator} can prevent this since $V^\mathsf{sum}_{\mathcal{M}_{i-1},m}(\boldsymbol{\ell})$  represents the increase in volume of our rectangular prism.
We thus aim to exploit both metrics for complexity and stability advantages for CRLB minimization.

Our resulting greedy sensor selection algorithm is summarized in Algorithm~\ref{alg:dynamic_fraction}.
The first sensor (i.e., $\mathcal{M}_1$) is randomly selected from $\mathcal{M}_\text{max}$, and for the rest of selection steps except for $i=3$, in which we use~\eqref{eq:dynamic_denominator}, sensors are selected based on~\eqref{eq:dynamic_numerator} using exhaustive search. 
Note that we use~\eqref{eq:dynamic_denominator} when $i=3$ to prevent the algorithm from selecting co-planar sensors.
\vspace{-3mm}
\begin{algorithm}[ht]
{\small
\caption{Greedy Sensor Selection based on~\eqref{eq:dynamic_numerator} and~\eqref{eq:dynamic_denominator}}
\label{alg:dynamic_fraction}
\begin{algorithmic}
\REQUIRE $\epsilon_m$ and $\boldsymbol{u}_m$, $\forall m \in \mathcal{M}_\text{max}$
\STATE Generate $\mathcal{M}_1:$ a randomly selected sensor from $\mathcal{M}_\text{max}$
\STATE $i = 2$
\WHILE {$i \leq M$}
    \IF {$i = 3$}
    \STATE Find $m_i^\star$ from solving~\eqref{eq:dynamic_denominator} via exhaustive search
    \ELSE
    \STATE Find $m_i^\star$ from solving~\eqref{eq:dynamic_numerator} via exhaustive search
    \ENDIF
    \STATE $\mathcal{M}_i = \mathcal{M}_{i-1} \cup \{m_i^\star\}$
    \STATE $i = i + 1$
\ENDWHILE
\RETURN $\mathcal{M}_M$ 
\end{algorithmic}
}
\end{algorithm}

\vspace{-6mm}
\subsection{Computational Complexity Analysis} \label{ssec:dynamic_complexity}

For each algorithm, we break down the complexity analysis into two separate parts: (i) the number of arithmetic operations required to compute the expressions, independent of $M$ and $M_\text{max}$, that are repeatedly evaluated by the algorithm, and (ii) the total number of times the algorithm evaluates each of these expressions to complete the selection.
The results are summarized in Table~\ref{tb:complexity} and discussed in the following.

For every selection step, Algorithm~\ref{alg:dynamic_trace} repeatedly computes~\eqref{eq:marginal_CRLB} with 67 arithmetic operations and finds the sensor that satisfies~\eqref{eq:dynamic_trace}.
To select $M$ out of $M_\text{max}$ sensors,~\eqref{eq:marginal_CRLB} is computed $\sum^{M}_{i=4}(M_\text{max}-i+1)$ times by the algorithm.
As a result, Algorithm~\ref{alg:dynamic_trace} computes $\sum^{M}_{i=4}67(M_\text{max}-i+1)$ arithmetic operations to select $M$ out of $M_\text{max}$ sensors. 

For Algorithm~\ref{alg:dynamic_fraction}, which relies on~\eqref{eq:dynamic_numerator} and~\eqref{eq:dynamic_denominator} to conduct greedy selection, $\sum^{M}_{i=2}(M_\text{max}-i+1)(i-1)-2(M_\text{max}-2)$ computations of $A_{m_1m}$ and $(M_\text{max}-2)$ computations of $V_{m_1m_2m}$ are required to complete the selection of $M$ sensors.
Since $A_{m_1m}$ and $V_{m_1m_2m}$ require 3 and 6 arithmetic operations, respectively, we see that $3\big[\sum^{M}_{i=2}(M_\text{max}-i+1)(i-1)-2(M_\text{max}-2)\big]+6(M_\text{max}-2)=\sum^{M}_{i=2}3(M_\text{max}-i+1)(i-1)$ arithmetic operations are required by the algorithm.

\begin{table*}[t]
\captionsetup{justification=centering, labelsep=newline}
\centering
\caption{Complexity comparison of the two proposed dynamic sensor selection algorithms. While Algorithm~\ref{alg:dynamic_trace} has lower asymptotic complexity, \\ it requires more arithmetic operations for small values of $M_\text{max}$.}
\label{tb:complexity}
\begin{tabular}{|c|c|c|c|c|c|} 
    \hline
    Algorithm & Computation to Repeat & Repetitions & Total Arithmetic Operations & Time Complexity \\
    \hline
    \ref{alg:dynamic_trace} & \eqref{eq:marginal_CRLB} & $\sum^{M}_{i=4}(M_\text{max}\hspace{-0.5mm}-\hspace{-0.5mm}i\hspace{-0.5mm}+\hspace{-0.5mm}1)$ & $\sum^{M}_{i=4}43(M_\text{max}\hspace{-0.5mm}-\hspace{-0.5mm}i\hspace{-0.5mm}+\hspace{-0.5mm}1)$ & $\mathcal{O}\big(M_\text{max}^2\big)$ \\ 
    \hline
    \multirow{2}{*}{\ref{alg:dynamic_fraction}} & $\;\;\;\;A_{m_1m}$ for $i\neq 3$ & $\sum^{M}_{i=2}(M_\text{max}\hspace{-0.5mm}-\hspace{-0.5mm}i\hspace{-0.5mm}+\hspace{-0.5mm}1)(i\hspace{-0.5mm}-\hspace{-0.5mm}1)\hspace{-0.5mm}-\hspace{-0.5mm}2(M_\text{max}\hspace{-0.5mm}-\hspace{-0.5mm}2)$ & \multirow{2}{*}{$\sum^{M}_{i=2}3(M_\text{max}\hspace{-0.5mm}-\hspace{-0.5mm}i\hspace{-0.5mm}+\hspace{-0.5mm}1)(i\hspace{-0.5mm}-\hspace{-0.5mm}1)$} & \multirow{2}{*}{$\mathcal{O}\big(M_\text{max}^3\big)$} \\
    & $V_{m_1m_2m}$ for $i = 3$ & for $A_{m_1m}$ and $(M_\text{max}\hspace{-0.5mm}-\hspace{-0.5mm}2)$ for $V_{m_1m_2m}$ & & \\
    \hline
\end{tabular}
\vspace{-2mm}
\end{table*}

For comparison, we introduce the best option filling (BOF) algorithm~\cite{Zhao19}, where the complete CRLB in~\eqref{eq:CRLB_hybrid_corr_1} is evaluated to conduct greedy selection.
In BOF algorithm, the first three sensors are randomly selected, and
for each selection step $4 \leq i \leq M$, the algorithm computes $\epsilon_m\boldsymbol{u}_m\boldsymbol{u}_m^\top$ in~\eqref{eq:CRLB_hybrid_corr_1} for $i$ times and takes the inverse of their sum to evaluate the CRLB.
This results in a total of $\sum^{M}_{i=4}(12+N_\text{inv})(M_\text{max}-i+1)i$ arithmetic operations from the BOF algorithm, where $N_\text{inv}$ is the number of arithmetic operations involved in the matrix inverse.
The complete process of BOF algorithm is summarized in Algorithm~\ref{alg:dynamic_BOF}.

\begin{algorithm}[ht]
{\small
\caption{Best option filling (BOF) sensor selection~\cite{Zhao19}}
\label{alg:dynamic_BOF}
\begin{algorithmic}
\REQUIRE $\epsilon_m$ and $\boldsymbol{u}_m$, $\forall m \in \mathcal{M}_\text{max}$
\STATE Generate $\mathcal{M}_3:$ 3 randomly selected sensors from $\mathcal{M}_\text{max}$
\STATE $i=4$
\WHILE {$i \leq M$}
    \STATE Find $m_i^\star=\argmin_{m\in\mathcal{M}_\text{max}\backslash \mathcal{M}_{i-1}} \sigma^2_{\mathcal{M}_{i-1}\cup\{m\}}(\boldsymbol{\ell})$ using the expression of $\sigma^2_{\mathcal{M}}(\boldsymbol{\ell})$ in (22) via exhaustive search
    \STATE $\mathcal{M}_i = \mathcal{M}_{i-1} \cup \{m_i^\star\}$
    \STATE $i = i + 1$
\ENDWHILE
\RETURN $\mathcal{M}_M$ 
\end{algorithmic}
}
\end{algorithm}

Obtaining the polynomial expression for each algorithm's complexity, the leading terms are found to be $67MM_\text{max}-33.5M^2$ and $1.5M^2M_\text{max}-M^3$ for Algorithms~\ref{alg:dynamic_trace} and~\ref{alg:dynamic_fraction}, respectively.
We see that Algorithm~\ref{alg:dynamic_fraction} has a higher degree but much smaller coefficients than Algorithm~\ref{alg:dynamic_trace}.
This implies that Algorithm~\ref{alg:dynamic_fraction} maintains lower complexity for small values of $M_\text{max}$ but surpasses Algorithm~\ref{alg:dynamic_trace} as $M_\text{max}$ increases.
With $M = M_\text{max}$, we find the asymptotic complexities become $\mathcal{O}(M_\text{max}^2)$ and $\mathcal{O}(M_\text{max}^3)$ for Algorithm~\ref{alg:dynamic_trace} and \ref{alg:dynamic_fraction}, respectively.
Note that the asymptotic complexity of BOF algorithm is $\mathcal{O}(M_\text{max}^3)$, which is same as Algorithm~\ref{alg:dynamic_fraction}.
Considering that sensor selection via SDR and exhaustive search have asymptotic complexities of $\mathcal{O}(M_\text{max}^{4.5})$ and $\mathcal{O}(M_\text{max}!)$, respectively, both algorithms we propose show the computational advantage for large-scale systems having a large number of sensors.

\vspace{-2mm}
\section{Robust Sensor Selection Strategies}\label{sec:robust}

In this section, we turn to the robust sensor selection problem $\boldsymbol{\mathcal{P}}_\text{R}$ which aims to select a group of sensors such that the worst-case CRLB is minimized.
After discussing the unreliability issue stemming from directly applying the conventional convex relaxation technique (Section~\ref{ssec:robust_convex}), we present three new sensor selection strategies (Section~\ref{ssec:robust_algorithm1},\ref{ssec:robust_algorithm2},\ref{ssec:robust_algorithm3}) that provide reliable solutions to the robust sensor selection problem.

\vspace{-3mm}
\subsection{Sensor Selection via Convex Relaxation}\label{ssec:robust_convex}

We first adopt convex relaxation~\cite{Joshi09}.
Using~\eqref{eq:CRLB_hybrid_corr_2} and a binary selection vector $\boldsymbol{b}=[b_1,b_2,\ldots,b_{M_\text{max}}]^\top$, where $b_m\in\{0,1\}$, problem~$\boldsymbol{\mathcal{P}}_\text{R}$ can be rewritten as
\begin{align}
    \boldsymbol{b}_\text{R}^\star=&\argmin_{\boldsymbol{b}}\max_{g\in \{1,\ldots,G\}}\;\textrm{tr}\left\{(\mathbf{U}_g\mathbf{E}_g\mathbf{B}\mathbf{U}_g^\top)^{-1}\right\} \label{eq:robust_minmax_matrix_prob} \\
    \text{s.t.}\;\;&\boldsymbol{1}^\top\boldsymbol{b}=M, \label{eq:robust_minmax_matrix_const1} \\
    & b_m \in \{0,1\}\;\;\text{for}\;\;m=1,\ldots,M_\text{max}, \label{eq:robust_minmax_matrix_const2}
\end{align}
where $\mathbf{B}=\text{diag}(\boldsymbol{b})$, $\mathbf{U}_g = [\boldsymbol{u}_{1g},\boldsymbol{u}_{2g},\ldots,\boldsymbol{u}_{M_\text{max}g}]$, and $\mathbf{E}_g=\text{diag}([\epsilon_{1g},\epsilon_{2g},\ldots,\epsilon_{M_\text{max}g}]^\top)$ for each location $g$. 
Relaxing the binary constraint~\eqref{eq:robust_minmax_matrix_const2} to a continuous selection vector $\boldsymbol{c}=[c_1,\ldots,c_{M_\text{max}}]^\top$, $0\leq c_m \leq 1$, we rewrite the problem as
\begin{align}
    (\widetilde{\boldsymbol{\mathcal{P}}}_\text{R}):~
    \boldsymbol{c}_\text{R}^\star=&\argmin_{\boldsymbol{c}}\max_{g \in \{1,\ldots,G\}}\;\textrm{tr}\left\{(\mathbf{U}_g\mathbf{E}_g\mathbf{C}\mathbf{U}_g^\top)^{-1}\right\} \label{eq:robust_minmax_relaxed_prob} \\
    \text{s.t.}\;\;&\boldsymbol{1}^\top\boldsymbol{c}=M, \label{eq:robust_minmax_relaxed_const1} \\
    & 0 \leq c_m \leq 1 \;\; \text{for} \;\; m=1,\ldots,M_\text{max}, \label{eq:robust_minmax_relaxed_const2}
\end{align}
where $\mathbf{C}=\text{diag}(\boldsymbol{c})$.
Still, the above problem cannot be solved directly using convex optimization because of the \emph{max} operation over the target location $\boldsymbol{\ell}^\mathsf{t}_g$ for $g \in \{1,\ldots,G\}$.

We subsequently transform the min-max problem $\widetilde{\boldsymbol{\mathcal{P}}}_\text{R}$ into an equivalent convex minimization problem by introducing a threshold variable $\gamma$ to represent the maximum CRLB allowed out of the $G$ locations in $\mathcal{L}$.
The problem is formulated as
\begin{align}
    (\widehat{\boldsymbol{\mathcal{P}}}_\text{R}):~&\boldsymbol{c}_\text{R}^\star=\argmin_{\boldsymbol{c},\gamma}\;\gamma \label{eq:robust_min_prob} \\
    \text{s.t.}\;\;&\boldsymbol{1}^\top\boldsymbol{c}=M, \label{eq:robust_min_prob_const1} \\
    &0 \leq c_m \leq 1 \;\; \text{for} \;\; m=1,\ldots,M_\text{max}, \label{eq:robust_min_prob_const2} \\ 
    &\textrm{tr}\left\{(\mathbf{U}_g\mathbf{E}_g\mathbf{C}\mathbf{U}_g^\top)^{-1}\right\}\leq\gamma \;\; \text{for} \;\; g = 1,\ldots,G. \label{eq:robust_min_prob_const3}
\end{align}
The formulation of $\widehat{\boldsymbol{\mathcal{P}}}_\text{R}$ allows us to use convex optimization techniques, e.g., interior-point methods~\cite{Nesterov14}, to find the solution.
However, in contrast to the convex-relaxed dynamic sensor selection problem~\cite{Zhao19}, $G$ additional constraints are imposed in~\eqref{eq:robust_min_prob_const3}, which adds additional complexity for solving~$\widehat{\boldsymbol{\mathcal{P}}}_\text{R}$.

The convex relaxation approach in~$\widehat{\boldsymbol{\mathcal{P}}}_\text{R}$ requires an extra step of determining the binary selection vector $\widehat{\boldsymbol{b}}_\text{R}^\star$ whenever the solution of $\widehat{\boldsymbol{\mathcal{P}}}_\text{R}$ assigns non-zero values to more than $M$ sensors, i.e., $\|\boldsymbol{c}_\text{R}^\star\|_0 > M$.
A simple heuristic~\cite{Joshi09} to determine $\widehat{\boldsymbol{b}}_\text{R}^\star$ would be selecting the $M$ sensors with the largest $c_{m}$ values.
However, this can result in a poor selection result if any of the essential sensors are discarded.
In particular, some sensors assigned with lower $c_m$ could be as important as the sensors with higher $c_m$ for avoiding ill-conditioned FIMs, which lead to extremely large CRLBs. 

\begin{figure}[!t]
    \centering
    \includegraphics[width=1\linewidth]{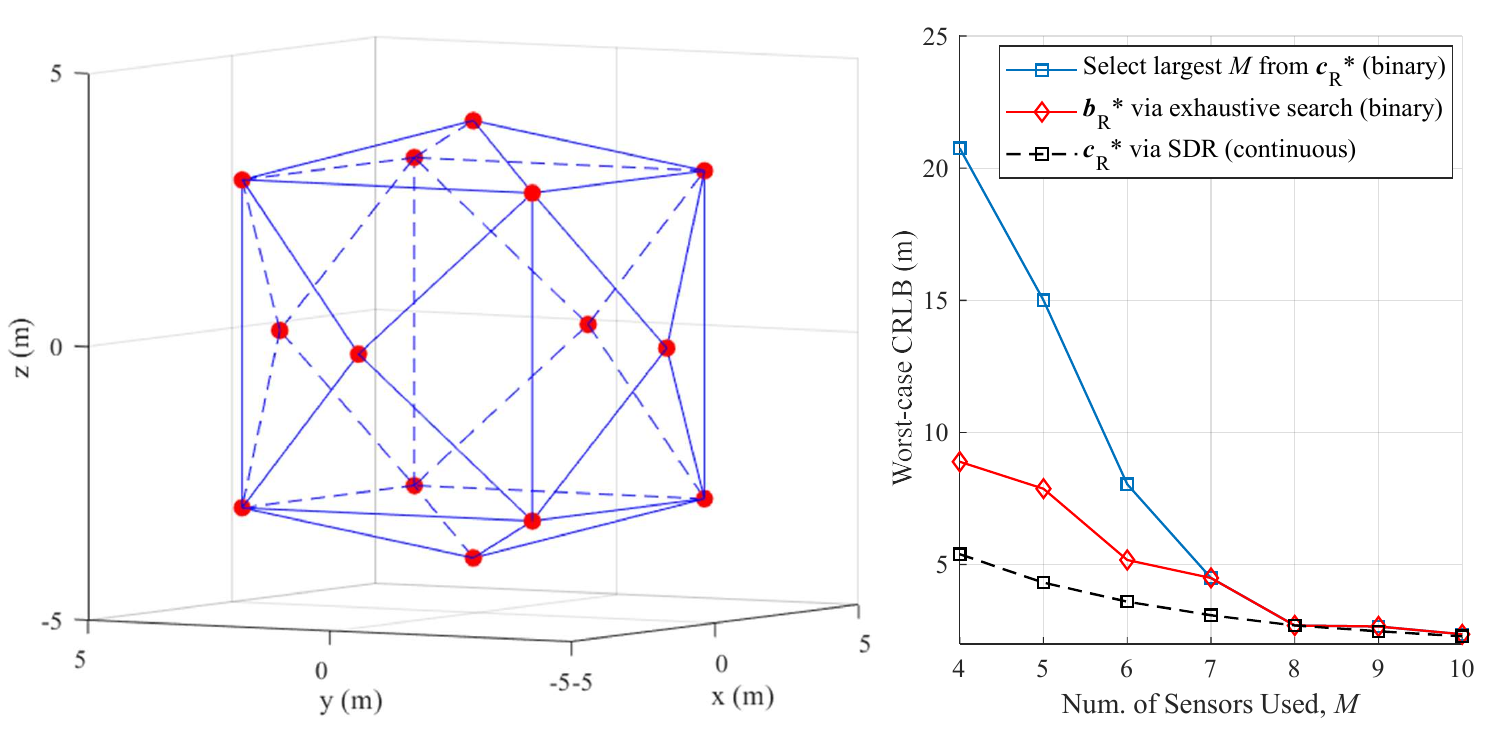}
    \caption{3D visual representation of 14 sensors in a prism shape with $d_\text{s}=4$ (left) and robust sensor selection performance comparison between binary and continuous selection cases (right).}
    \label{fig:Prism_and_Plot}
    \vspace{-2mm}
\end{figure}

To illustrate this phenomenon, we directly compare the CRLBs obtained by (i) $\boldsymbol{b}_\text{R}^\star$ found via exhaustive search and (ii) $\widehat{\boldsymbol{b}}_\text{R}^\star$ found by picking the largest $M$ sensors from $\boldsymbol{c}_\text{R}^\star$ in~\eqref{eq:robust_min_prob}.
We consider a sensor selection scenario with $M_\text{max}=14$ sensors arranged as shown in the left plot of Fig.~\ref{fig:Prism_and_Plot}.
The target location set $\mathcal{L}$ is generated to cover the space defined by $d_\text{s}=4$ m and $d_\text{max}=14$ m.
For each sensor $m$, we assume the TOA and RSS measurement noises with variances of $\sigma^2_{\text{T},m}$  and $\sigma^2_{\text{R},m}$, respectively, the values of which are selected the same as those described in Section~\ref{ssec:simulations_setup}.
The right plot of Fig.~\ref{fig:Prism_and_Plot} shows the result obtained over different values of $M$.
A significant performance gap exists between $\boldsymbol{b}_\text{R}^\star$ (red) and $\widehat{\boldsymbol{b}}_\text{R}^\star$ (blue), which indicates the latter approach is discarding important sensors.
The solution $\boldsymbol{c}^\star_\text{R}$ itself (black) indeed provides the lowest CRLB because, without the binary constraint, $\boldsymbol{c}_\text{R}^\star$ is allowed to be non-sparse and select as many sensors as needed.
However, we see that the effectiveness of $\boldsymbol{c}^\star_\text{R}$ does not guarantee that the largest $M$ sensors chosen based on it will be close to optimal.

Such difference between these two binary solutions implies that relying on the convex relaxation technique and naively manipulating $\boldsymbol{c}_\text{R}^\star$ is a poor strategy for robust sensor selection.
In the following, we propose three distinct approaches to conduct robust sensor selection.

\vspace{-3mm}
\subsection{Sensor Selection via Iterative Convex Optimization}\label{ssec:robust_algorithm1}

Our first strategy, \emph{iterative convex optimization} (ICO), is a greedy selection based on consecutive iterations of convex optimization.
In each optimization iteration, a sensor that is found to be the most beneficial for minimizing the worst-case CRLB is selected.
Defining $\mathcal{M}_\text{sel}$ to be the set of sensors selected by the algorithm, initially $\mathcal{M}_\text{sel}=\varnothing$, for every iteration step, the algorithm solves the following convex optimization problem:
\begin{align}
    (\widehat{\boldsymbol{\mathcal{P}}}_\text{R1}):~&\boldsymbol{c}_\text{R1}^\star=\min_{\boldsymbol{c},\gamma}\; \gamma \label{eq:robust_min_iter_prob} \\
    \text{s.t.}\;\;& \boldsymbol{1}^\top\boldsymbol{c}=M, \label{eq:robust_min_iter_const1} \\
    & 0 \leq c_m \leq 1 \;\; \text{for} \;\; m\in\mathcal{M}_\text{max} \backslash \mathcal{M}_\text{sel}, \label{eq:robust_min_iter_const2} \\
    & c_m = 1 \;\; \text{for} \;\; m\in\mathcal{M}_\text{sel}, \label{eq:robust_min_iter_const3} \\
    &\textrm{tr}\left\{(\mathbf{U}_g\mathbf{E}_g\mathbf{C}\mathbf{U}_g^\top)^{-1}\right\}\leq\gamma \;\; \text{for} \;\; g = 1,\ldots,G. \label{eq:robust_min_iter_const4}
\end{align}
In each step, the optimization is performed over the $M_\text{max}-\vert \mathcal{M}_\text{sel}\vert $ sensors that have not yet been selected.
Among the sensors in $\mathcal{M}_\text{max} \backslash \mathcal{M}_\text{sel}$, the one with the largest $c_m$ is added to $\mathcal{M}_\text{sel}$.
The iterations continue until all $M$ sensors are selected.
The overall procedure is summarized in Algorithm~\ref{alg:iter-CVX}.

\begin{algorithm}[ht]
{\small
\caption{Robust Sensor Selection based on ICO}
\label{alg:iter-CVX}
\begin{algorithmic}
\REQUIRE $\epsilon_{mg}$ and $\boldsymbol{u}_{mg}$, $\forall m \in \mathcal{M}_\text{max}$ and $\forall g \in \mathcal{L}$
\STATE Initialize $\mathcal{M}_\text{sel} = \varnothing$, $\boldsymbol{b} = \boldsymbol{0}$
\STATE $i = 1$
\WHILE {$i \leq M$}
    \STATE Acquire $\boldsymbol{c}_\text{R1}^\star$ from solving $\widehat{\boldsymbol{\mathcal{P}}}_\text{R1}$ via convex optimization
    \STATE $m^\star = \argmax_{m\in\mathcal{M}_\text{max} \setminus \mathcal{M}_\text{sel}} \boldsymbol{c}_\text{R1}^\star[m]$
    \STATE $b_{m^\star} = 1$
    \STATE $\mathcal{M}_\text{sel} = \mathcal{M}_\text{sel} \cup \{m^\star\}$
    \STATE $i = i + 1$
\ENDWHILE
\RETURN $\boldsymbol{b}$ 
\end{algorithmic}
}
\end{algorithm}

This greedy selection strategy promotes robustness because each sensor is sequentially and individually selected such that $\mathcal{M}_\text{sel}$ focuses on maximizing performance without relying on any unselected sensors.
In particular, the convex optimization of each selection step reflects the selections from earlier steps via~\eqref{eq:robust_min_iter_const3}.
Conducting greedy selection based on these per-step solutions provides us with a stable and yet effective sensor selection in the end.
Note that the mechanism of our strategy is similar to matching pursuit~\cite{Mallat93}, a widely adopted solution for NP-hard sparse approximation problems, in the sense that the algorithm primarily looks for the sensors that have the biggest impact on the worst-case CRLB minimization.
Nevertheless, the suboptimality of greedy selection can limit the effectiveness of Algorithm~\ref{alg:iter-CVX}. In the following strategy, we propose a strategy that does not employ the greedy selection framework.

\vspace{-3mm}
\subsection{Sensor Selection via Difference of Convex Functions Programming}\label{ssec:robust_algorithm2}

Instead of finding the entries of $\boldsymbol{b}_\text{R}^\star$ one at a time, here we formulate an optimization problem such that elements of $\boldsymbol{c}_\text{R}^\star$ are forced to be binary.
We can achieve this in theory by adding the constraint $\boldsymbol{1}^\top\boldsymbol{c} - \boldsymbol{c}^\top\boldsymbol{c} \leq 0$ to $\widehat{\boldsymbol{\mathcal{P}}}_\text{R}$, which, together with~\eqref{eq:robust_min_prob_const2}, constrains $c_m\in\{0,1\}$.
However, this would make our optimization problem non-convex.
We thus impose this by adding a penalty term to the objective function for penalizing non-binary solutions.
We formally reformulate our robust sensor selection problem as
\begin{align}
    (\widehat{\boldsymbol{\mathcal{P}}}_\text{R2}):~&\boldsymbol{c}_\text{R2}^\star=\argmin_{\boldsymbol{c},\gamma}\;\gamma + \lambda(\boldsymbol{1}^\top\boldsymbol{c} - \boldsymbol{c}^\top\boldsymbol{c}) \label{eq:robust_min_DC_prob} \\
    \text{s.t.}\;\;&\boldsymbol{1}^\top\boldsymbol{c}=M, \label{eq:robust_min_DC_const1} \\
    &0 \leq c_m \leq 1 \;\; \text{for} \;\; m=1,\ldots,M_\text{max}, \label{eq:robust_min_DC_const2} \\ 
    &\textrm{tr}\left\{(\mathbf{U}_g\mathbf{E}_g\mathbf{C}\mathbf{U}_g^\top)^{-1}\right\}\leq\gamma \;\; \text{for} \;\; g = 1,\ldots,G, \label{eq:robust_min_DC_const3}
\end{align}
where $\lambda \geq 0$ is the penalty factor.
Within the range set by~\eqref{eq:robust_min_DC_const2}, each $c_m$ would take binary entries to have the quadratic penalty term added to the objective function decreased.
However, introducing this concave penalty term still makes our problem non-convex.
To circumvent this issue, we exploit an optimization technique called \emph{difference of convex functions programming} (DCP)~\cite{Lipp16} to solve $\widehat{\boldsymbol{\mathcal{P}}}_\text{R2}$.

For DCP, we write the objective function in~\eqref{eq:robust_min_DC_prob} as
\begin{equation}
    f(\boldsymbol{c}) - g(\boldsymbol{c}), \label{eq:DC_form}
\end{equation}
where $f(\boldsymbol{c})= \gamma + \lambda\boldsymbol{1}^\top\boldsymbol{c}$ and $g(\boldsymbol{c})=\lambda\boldsymbol{c}^\top\boldsymbol{c}$.
Note that $f(\boldsymbol{c})$ and $g(\boldsymbol{c})$ are both convex functions with respect to the optimization variables.
Our objective function is thus the difference of two convex functions of $\boldsymbol{c}$, as DCP requires.
Then, to make $g(\boldsymbol{c})$ affine, we apply a linear approximation to $g(\boldsymbol{c})$~\cite{Lipp16} as
\begin{align}
    \widetilde{g}(\boldsymbol{c};\boldsymbol{c}_k) &= g(\boldsymbol{c}_k)+\nabla g(\boldsymbol{c}_k)^\top(\boldsymbol{c}-\boldsymbol{c}_k) \label{eq:DC_approximation1} \\
    & = \lambda\boldsymbol{c}_k^\top\boldsymbol{c}_k + 2\lambda\boldsymbol{c}_k^\top(\boldsymbol{c}-\boldsymbol{c}_k), \label{eq:DC_approximation2}
\end{align}
where $\boldsymbol{c}_k$ is a feasible point to $\widehat{\boldsymbol{\mathcal{P}}}_\text{R2}$.
By the first order condition of a convex function $g(\boldsymbol{c})$ at $\boldsymbol{c}$, i.e., $g(\boldsymbol{c}) \geq \tilde g(\boldsymbol{c}; \boldsymbol{c}_k)$, DCP operates as follows.
At iteration step $k$, the objective function of $\widehat{\boldsymbol{\mathcal{P}}}_\text{R2}$, i.e.,~\eqref{eq:DC_form}, is replaced by $f(\boldsymbol{c}) - \widetilde{g}(\boldsymbol{c};\boldsymbol{c}_k)$, and we use convex optimization, e.g., an interior-point method~\cite{Nesterov14}, to find the solution, denoted by $\boldsymbol{c}_k^\star$.
For the next iteration step $k+1$, the same objective function is replaced by $f(\boldsymbol{c}) - \widetilde{g}(\boldsymbol{c};\boldsymbol{c}_{k+1})$, where we set $\boldsymbol{c}_{k+1}=\boldsymbol{c}_k^\star$, and the solution $\boldsymbol{c}_{k+1}^\star$ is once again obtained via convex optimization.
The steps are repeated until the solution converges, i.e., $\|\boldsymbol{c}^\star_{k+1}-\boldsymbol{c}^\star_{k}\| < \varepsilon$.

We next show that a solution obtained by DCP is indeed a stationary point in $\widehat{\boldsymbol{\mathcal{P}}}_\text{R2}$ through the following proposition.
\begin{proposition}
    If there exists a feasible point $\boldsymbol{c}^\star$ to which DCP converges (i.e., $\boldsymbol{c}_{k}^\star\rightarrow \boldsymbol{c}^\star$ as $k \rightarrow \infty$), $\boldsymbol{c}^\star$ is a stationary point to $\widehat{\boldsymbol{\mathcal{P}}}_\text{R2}$. \label{prop:DC_stationary}
\end{proposition}
\begin{proof}
    First, note that $\boldsymbol{c}_k^\star$ obtained by DCP is feasible to $\widehat{\boldsymbol{\mathcal{P}}}_\text{R2}$ for all $k$ since there are no approximations made to any of the constraints in $\widehat{\boldsymbol{\mathcal{P}}}_\text{R2}$.
    
    We next show that the objective function of $\widehat{\boldsymbol{\mathcal{P}}}_\text{R2}$ at $\boldsymbol{c}_k^\star$, i.e., $f(\boldsymbol{c}_k^\star) - g(\boldsymbol{c}_k^\star)$, is bounded by the objective functions of DCP from two consecutive steps $k$ and $k+1$.
    The relationship between the objective values of DCP evaluated at $\boldsymbol{c}_{k}^\star$ and $\boldsymbol{c}_{k+1}^\star$, obtained from iteration steps $k$ and $k+1$, is given by~\cite{Lipp16}
    \begin{align}
    f(\boldsymbol{c}_{k}^\star)-\widetilde{g}(\boldsymbol{c}_{k}^\star;\boldsymbol{c}_{k}) &\geq f(\boldsymbol{c}_k^\star)-g(\boldsymbol{c}_k^\star) \label{proof:DC_nonincreasing1} \\
    &=f(\boldsymbol{c}_{k+1})-g(\boldsymbol{c}_{k+1}) \label{proof:DC_nonincreasing2} \\
    &=f(\boldsymbol{c}_{k+1})-\widetilde{g}(\boldsymbol{c}_{k+1};\boldsymbol{c}_{k+1}) \label{proof:DC_nonincreasing3} \\
    &\geq f(\boldsymbol{c}_{k+1}^\star)-\widetilde{g}(\boldsymbol{c}_{k+1}^\star;\boldsymbol{c}_{k+1}).\label{proof:DC_nonincreasing4}
    \end{align}
    The equality in~\eqref{proof:DC_nonincreasing2} holds from the operation of DCP, and the equality in~\eqref{proof:DC_nonincreasing3} is true since $\widetilde{g}(\boldsymbol{c};\boldsymbol{c}_{k+1})$ is the linear approximation of $g(\boldsymbol{c})$.
    The fact that $\boldsymbol{c}_{k+1}^\star$ minimizes $f(\boldsymbol{c})-\widetilde{g}(\boldsymbol{c};\boldsymbol{c}_{k+1})$ establishes the last inequality.
    By the relationship derived above, $f(\boldsymbol{c}_k^\star) - g(\boldsymbol{c}_k^\star)$ is lower-bounded by $f(\boldsymbol{c}_{k+1}^\star)-\widetilde{g}(\boldsymbol{c}_{k+1}^\star;\boldsymbol{c}_{k+1})$ and also upper-bounded by $f(\boldsymbol{c}_{k}^\star)-\widetilde{g}(\boldsymbol{c}_{k}^\star;\boldsymbol{c}_{k})$.
    
    If DCP converges to a point $\boldsymbol{c}^\star$, it implies that $\boldsymbol{c}_k^\star=\boldsymbol{c}_{k+1}^\star=\boldsymbol{c}^\star$ as $k\rightarrow \infty$.
    As a result, the objective function of $\widehat{\boldsymbol{\mathcal{P}}}_\text{R2}$, which is both lower and upper bounded by the same function $f(\boldsymbol{c}^\star)-\widetilde{g}(\boldsymbol{c}^\star;\boldsymbol{c}^\star)$, converges to a stationary point $f(\boldsymbol{c}^\star) - g(\boldsymbol{c}^\star)$.
\end{proof}

With the penalty term introduced to force binary solutions for $\widehat{\boldsymbol{\mathcal{P}}}_\text{R2}$, multiple locally optimal points can exist in the feasible space.
Therefore, despite the effectiveness of DCP in finding stationary points of $\widehat{\boldsymbol{\mathcal{P}}}_\text{R2}$, no guarantee is given on $\boldsymbol{c}_\text{R2}^\star=\boldsymbol{b}_\text{R}^\star$.

Instead, the ability of DCP to find the closest $\boldsymbol{c}_\text{R2}^\star$ to $\boldsymbol{b}_\text{R}^\star$ depends on the choice of initial point, i.e., $\boldsymbol{c}_0$~\cite{Lipp16}.
One heuristic to find such solution is therefore to run DCP multiple times with different $\boldsymbol{c}_0$ values and select the solution with the lowest objective value.
The required number of runs for this method, however, depends on the magnitude of $\lambda$ in $\widehat{\boldsymbol{\mathcal{P}}}_\text{R2}$.
On the one hand, with $\lambda = 0$, DCP finds the same non-binary solution regardless of $\boldsymbol{c}_0$, so only a single run of DCP is required.
On the other hand, with $\lambda = \infty$, the solution is strictly determined by $\boldsymbol{c}_0$, so the required number of runs may increase up to ${M_\text{max}}\choose{M}$, which is equivalent to the exhaustive search case.
Therefore, both $\lambda$ and the number of runs, which we denote as $N_\text{DCP}$, must be carefully selected.

We propose setting $\lambda=\kappa\gamma_0$, where $\gamma_0$ is the optimized $\gamma$ from solving $\widehat{\boldsymbol{\mathcal{P}}}_\text{R2}$ with $\lambda =0$, and $\kappa \geq 0$ is a scaling factor.
By setting $\lambda$ to be proportional to $\gamma_0$, we can balance our penalty term based on the worst-case CRLB (i.e, $\gamma$ in $\widehat{\boldsymbol{\mathcal{P}}}_\text{R2}$) and effectively force our optimization to find a desired binary solution.
The overall procedure of DCP-based robust sensor selection is summarized in Algorithm~\ref{alg:DC}.

\begin{algorithm}[ht]
{\small
\caption{Robust Sensor Selection based on DCP}
\label{alg:DC}
\begin{algorithmic}
\REQUIRE $N_\text{DCP}$, $\kappa$, $\varepsilon$, $\epsilon_m$ and $\boldsymbol{u}_m$, $\forall m \in \mathcal{M}_\text{max}$ 
\STATE Initialize $\mathcal{C}=\varnothing$, $\lambda = 0$ and generate a random feasible point $\boldsymbol{c}_0$
\STATE Acquire $\boldsymbol{c}_\text{R2}^\star$ from solving $\widehat{\boldsymbol{\mathcal{P}}}_\text{R2}$ via convex optimization
\STATE $\lambda = \kappa\gamma$ and $n=1$
\WHILE {$n \leq N_\text{DCP}$}
    \STATE $k=0$ and generate a random feasible point $\boldsymbol{c}_k$
    \WHILE {true}
        \STATE Update the objective function of $\widehat{\boldsymbol{\mathcal{P}}}_\text{R2}$ to $f(\boldsymbol{c})-\widetilde{g}(\boldsymbol{c};\boldsymbol{c}_k)$
        \STATE Acquire $\boldsymbol{c}_\text{R2}^\star$ from solving $\widehat{\boldsymbol{\mathcal{P}}}_\text{R2}$ via convex optimization
        \STATE $\boldsymbol{c}_{k+1} = \boldsymbol{c}_{\text{R2}}^\star$
        \IF {$\|\boldsymbol{c}_{k+1}-\boldsymbol{c}_{k}\| < \varepsilon$}
            \STATE Break
        \ENDIF
        \STATE $k = k+1$
    \ENDWHILE
    \STATE $\mathcal{C} = \mathcal{C} \cup \{\boldsymbol{c}_\text{R2}^\star\}$ and $n = n + 1$
\ENDWHILE
\STATE $\boldsymbol{c}_\text{R2}^{\star\star} = \argmin_{\boldsymbol{c}\in\mathcal{C}}f(\boldsymbol{c})-g(\boldsymbol{c})$
\RETURN $\boldsymbol{c}_\text{R2}^{\star\star}$ 
\end{algorithmic}
}
\end{algorithm}

To illustrate the impact of $\lambda$ on our proposed algorithm, we conduct a simulation to evaluate both the worst-case CRLB and zero-penalty rate (i.e., the rate with which our algorithm converges to a binary solution and yields zero penalty) over different values of $\kappa$.
For the simulation, we used the same system setup for generating Fig.~\ref{fig:Prism_and_Plot} except that sensors are randomly placed.
With the DCP-related parameters set as $N_\text{DCP}=20$, $\kappa \in [0.2,5]$, and $\varepsilon=0.05$, the obtained result is shown in Fig.~\ref{fig:DC_analysis}.
We see that higher $\kappa$ results in higher rate of finding binary solutions as greater penalty further ensures that DCP converges to a binary solution.
However, the worst-case CRLB performance also moves further away from the optimal exhaustive search case; imposing too much penalty on the optimization is less likely to provide a desired solution due to the increased chance of DCP finding one of the locally optimal points.
It is therefore important to determine the proper value of $\lambda$ so that the best performance is achieved from a given sensor selection scenario.
For example, from the upper plot of Fig.~\ref{fig:DC_analysis}, we find that our algorithm obtains the best performance with $\kappa = 0.5$ for $M=4$ but $\kappa = 0.2$ for $M=5,6$.

\begin{figure}[t]
    \vspace{-2mm}
    \centering
    \includegraphics[width=1\linewidth]{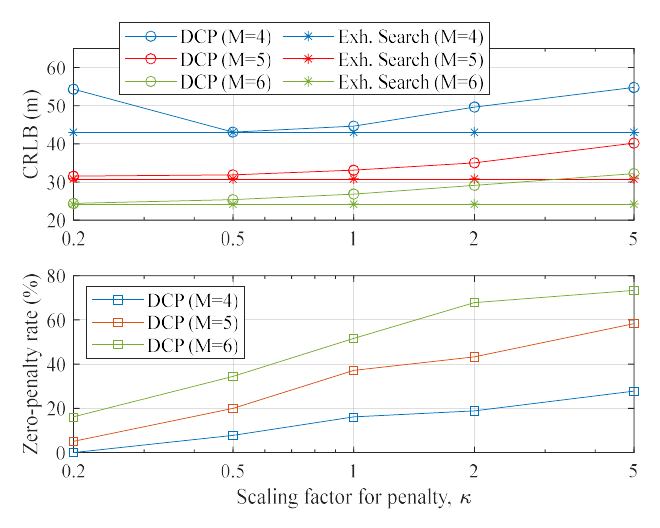}
    \caption{Worst-case CRLBs and zero-penalty rates over different values of $\kappa$.}
    \vspace{-3mm}
    \label{fig:DC_analysis}
\end{figure}

Our DCP-based strategy does not rely on greedy selection, but the necessity of finding a proper value for $\lambda$ and uncertainty of acquiring optimal solutions render this strategy less effective. Therefore, in the following, we introduce a more extreme strategy that is guaranteed to find the optimal set of sensors.

\vspace{-3mm}
\subsection{Sensor Selection via Discrete Monotonic Optimization}\label{ssec:robust_algorithm3}

Our final strategy is based on \emph{discrete monotonic optimization} (DMO)~\cite{Tuy06}, an optimization method for when the objective function and constraints are in the form of a difference of two monotonic functions.
In other words, DMO can be exploited for problems of the form
\begin{align}
    \max_{\boldsymbol{c}}\; &f^+(\boldsymbol{c}) - f^-(\boldsymbol{c}) \label{eq:DMO_obj} \\
    \text{s.t.}\;\; & g_q(\boldsymbol{c}) - h_q(\boldsymbol{c}) \leq 0 \;\; \text{for } q=1,\ldots,Q, \label{eq:DMO_const1}\\
    &\boldsymbol{c} \in [\boldsymbol{v}, \boldsymbol{w}], \label{eq:DMO_const2}
\end{align}
where $f^+(\boldsymbol{c})$, $f^-(\boldsymbol{c})$, $g_q(\boldsymbol{c})$, and $h_q(\boldsymbol{c})$ are all monotonically increasing\footnote{We consider a function $f$ to be \emph{increasing} in $\boldsymbol{c}$ if $f(\boldsymbol{c}^{(2)}) \geq f(\boldsymbol{c}^{(1)})$ for $\boldsymbol{c}^{(2)} \geq \boldsymbol{c}^{(1)}$, where $\boldsymbol{c}^{(2)} \geq \boldsymbol{c}^{(1)}$ implies $c^{(2)}_m \geq c^{(1)}_m$, $\forall m$.} functions in $\boldsymbol{c}$.
The \emph{Branch-Reduce-and-Bound} (BRB) technique~\cite{Tuy06,Kim15,Duan20} can be applied to solve DMO problems.
DMO is known to solve combinatorial optimization problems optimally because, unlike convex optimization methods, it handles the binary constraint without any relaxation~\cite{Tuy06}.

We propose the following re-formulation of~\eqref{eq:robust_minmax_relaxed_prob}-\eqref{eq:robust_minmax_relaxed_const2} for robust sensor selection:  
\begin{align}
    (\widetilde{\boldsymbol{\mathcal{P}}}_\text{R3}):~\boldsymbol{c}_\text{R3}^\star=&\argmax_{\boldsymbol{c}}\min_{g\in \{1,\ldots,G\}}-\textrm{tr}\left\{(\mathbf{U}_g\mathbf{E}_g\mathbf{C}\mathbf{U}_g^\top)^{-1}\right\} \nonumber \\ &\quad\quad\quad\quad\quad\quad-\mu\max(0,\boldsymbol{1}^\top\boldsymbol{c}-M) \label{eq:robust_minmax_DMO_prob} \\
    \text{s.t.}\;\; & \boldsymbol{1}^\top\boldsymbol{c} - \boldsymbol{c}^\top\boldsymbol{c} \leq 0, \label{eq:robust_minmax_DMO_const1} \\
    & \boldsymbol{c} \in [\boldsymbol{0},\boldsymbol{1}], \label{eq:robust_minmax_DMO_const2}
\end{align}
from which we can make the following list of mapping from~\eqref{eq:robust_minmax_DMO_prob}-\eqref{eq:robust_minmax_DMO_const2} to~\eqref{eq:DMO_obj}-\eqref{eq:DMO_const2}:
\begin{align}
    f^+(\boldsymbol{c}) &= \min_{g\in \{1,\ldots,G\}}-\textrm{tr}\left\{(\mathbf{U}_g\mathbf{E}_g\mathbf{C}\mathbf{U}_g^\top)^{-1}\right\}, \label{DMO_f+}\vspace{-1mm}\\
    f^-(\boldsymbol{c}) &= \mu\max(0,\boldsymbol{1}^\top\boldsymbol{c}-M), \label{DMO_f-} \\
    g_1(\boldsymbol{c}) &= \boldsymbol{1}^\top\boldsymbol{c},\;\;h_1(\boldsymbol{c}) = \boldsymbol{c}^\top\boldsymbol{c},\;\;[\boldsymbol{v},\boldsymbol{w}] = [\boldsymbol{0},\boldsymbol{1}].
\end{align}
In $\widetilde{\boldsymbol{\mathcal{P}}}_\text{R3}$, we have converted the constraint~\eqref{eq:robust_minmax_relaxed_const1} to a penalty function $\mu\max(0,\boldsymbol{1}^\top\boldsymbol{c}-M)$ where the \emph{max} operation with zero is applied to ensure $\boldsymbol{1}^\top\boldsymbol{c}=M$.
Note that the two constraints~\eqref{eq:robust_minmax_DMO_const1} and~\eqref{eq:robust_minmax_DMO_const2} enforce $\boldsymbol{c}$ to be binary.
Clearly, $f^-(\boldsymbol{c})$, $g_1(\boldsymbol{c})$, and $h_1(\boldsymbol{c})$ are increasing functions in $\boldsymbol{c}\in [\boldsymbol{0},\boldsymbol{1}]$.
This behavior holds for $f^{+}(\boldsymbol{c})$ as well:
\begin{proposition}
    $f^+(\boldsymbol{c})$ in~\eqref{DMO_f+} is an increasing function in $\boldsymbol{c}\in [\boldsymbol{0},\boldsymbol{1}]$. \label{prop:DMO_increasing}
\end{proposition}
\begin{proof}
    Consider two selection vectors $\boldsymbol{c}^{(1)}, \boldsymbol{c}^{(2)} \in [\boldsymbol{0},\boldsymbol{1}]$ such that $\boldsymbol{c}^{(2)} \geq \boldsymbol{c}^{(1)}$.
    Also define $\mathcal{M}^{(1)}=\{m\vert m\in\mathcal{M}_\text{max}, c_m^{(1)}>0\}$ and $\mathcal{M}^{(2)}=\{m\vert m\in\mathcal{M}_\text{max}, c_m^{(2)}>0\}$ to be the sets of sensors selected by $\boldsymbol{c}^{(1)}$ and $\boldsymbol{c}^{(2)}$, respectively.
    We identify two distinct cases of $\boldsymbol{c}^{(2)} \geq \boldsymbol{c}^{(1)}$ that can appear in sensor selection: (i) $\mathcal{M}^{(1)} = \mathcal{M}^{(2)}$~and (ii) $\mathcal{M}^{(1)} \subset \mathcal{M}^{(2)}$.
    For both cases, we show that $f^+(\boldsymbol{c}^{(2)})\geq f^+(\boldsymbol{c}^{(1)})$ if $\boldsymbol{c}^{(2)} \geq \boldsymbol{c}^{(1)}$.
    
    (i) $\mathcal{M}^{(1)} = \mathcal{M}^{(2)}$:
    Since the same sensors have been selected by both $\boldsymbol{c}^{(1)}$ and $\boldsymbol{c}^{(2)}$, the \emph{geometric conditionings}~\cite{Patwari03} of $\mathcal{M}^{(1)}$ and $\mathcal{M}^{(2)}$ are the same for each $g$.
    Since $g_\star^{(1)}$ $= \argmin_{g} -\textrm{tr}\left\{(\mathbf{U}_g\mathbf{E}_g\mathbf{C}^{(1)}\mathbf{U}_g)^{-1}\right\}$~and~$g_\star^{(2)}$ $= \argmin_{g} -\textrm{tr}\left\{(\mathbf{U}_g\mathbf{E}_g\mathbf{C}^{(2)}\mathbf{U}_g)^{-1}\right\}$ are equal, i.e., $g_\star$ $=g_\star^{(1)}$ $=g_\star^{(2)}$, we directly compare $-\textrm{tr}\left\{(\mathbf{U}_{g_\star}\mathbf{E}_{g_\star}\mathbf{C}^{(1)}\mathbf{U}_{g_\star})^{-1}\right\}$ with $-\textrm{tr}\left\{(\mathbf{U}_{g_\star}\mathbf{E}_{g_\star}\mathbf{C}^{(2)}\mathbf{U}_{g_\star})^{-1}\right\}$.
    Define $\mathbf{A}=\mathbf{U}_{g_\star}\mathbf{E}_{g_\star}\mathbf{C}^{(1)}\mathbf{U}_{g_\star}$ and $\mathbf{B}=\mathbf{U}_{g_\star}\mathbf{E}_{g_\star}(\mathbf{C}^{(2)}-\mathbf{C}^{(1)})\mathbf{U}_{g_\star}$ so that $\mathbf{A}+\mathbf{B}=\mathbf{U}_{g_\star}\mathbf{E}_{g_\star}\mathbf{C}^{(2)}\mathbf{U}_{g_\star}$.
    Then, from the Woodbury identity, we have $-\textrm{tr}\{(\mathbf{A}+\mathbf{B})^{-1}\}=-\textrm{tr}\{\mathbf{A}^{-1}\}+\textrm{tr}\{\mathbf{A}^{-1}(\mathbf{B}^{-1}+\mathbf{A}^{-1})^{-1}\mathbf{A}^{-1}\}$.
    Since both $\mathbf{A}$ and $\mathbf{B}$ are positive semidefinite, $\textrm{tr}\{\mathbf{A}^{-1}(\mathbf{B}^{-1}+\mathbf{A}^{-1})^{-1}\mathbf{A}^{-1}\}$ must be non-negative.
    Therefore, $-\textrm{tr}\left\{(\mathbf{U}_{g_\star}\mathbf{E}_{g_\star}\mathbf{C}^{(2)}\mathbf{U}_{g_\star})^{-1}\right\}\geq -\textrm{tr}\left\{(\mathbf{U}_{g_\star}\mathbf{E}_{g_\star}\mathbf{C}^{(1)}\mathbf{U}_{g_\star})^{-1}\right\}$.
    
    (ii) $\mathcal{M}^{(1)} \subset \mathcal{M}^{(2)}$:
    Let $m'$ be the index of sensors only selected by $\mathcal{M}^{(2)}$, i.e., $m'\in\mathcal{M}^{(2)}\backslash\mathcal{M}^{(1)}$.
    The relationship between the worst-case CRLBs resulting from $\boldsymbol{c}^{(1)}$ and $\boldsymbol{c}^{(2)}$ is given by
    \begin{align}
        -\textrm{tr}&\left\{(\mathbf{U}_{g_\star^{(2)}}\mathbf{E}_{g_\star^{(2)}}\mathbf{C}^{(2)}\mathbf{U}_{g_\star^{(2)}})^{-1}\right\} \nonumber \\
        &\geq \lim_{c_{m'}\rightarrow 0,\forall m'}-\textrm{tr}\left\{(\mathbf{U}_{g_\star^{(2)}}\mathbf{E}_{g_\star^{(2)}}\mathbf{C}^{(2)}\mathbf{U}_{g_\star^{(2)}})^{-1}\right\} \label{eq:DMO_increasing_proof_1} \\
        &\geq -\textrm{tr}\left\{(\mathbf{U}_{g_\star^{(2)}}\mathbf{E}_{g_\star^{(2)}}\mathbf{C}^{(1)}\mathbf{U}_{g_\star^{(2)}})^{-1}\right\} \label{eq:DMO_increasing_proof_2} \\ 
        &\geq -\textrm{tr}\left\{(\mathbf{U}_{g_\star^{(1)}}\mathbf{E}_{g_\star^{(1)}}\mathbf{C}^{(1)}\mathbf{U}_{g_\star^{(1)}})^{-1}\right\}, \label{eq:DMO_increasing_proof_3}
    \end{align}
    where~\eqref{eq:DMO_increasing_proof_1}-\eqref{eq:DMO_increasing_proof_2} are due to $\boldsymbol{c}^{(2)}\geq\lim_{c_{m'}\rightarrow 0,\forall m'} \boldsymbol{c}^{(2)} \geq \boldsymbol{c}^{(1)}$. The last inequality follows from the fact that the CRLB with $\boldsymbol{c}^{(1)}$ achieves its maximum at the $g_\star^{(1)}$ location.
\end{proof}

With $\widetilde{\boldsymbol{\mathcal{P}}}_\text{R3}$ formulated, our DMO-based robust sensor selection algorithm has the following branch and reduce steps:
\begin{itemize}[leftmargin=4mm]
    \item {\tt Branch}: A box $B=[\boldsymbol{v},\boldsymbol{w}]$ is partitioned into two boxes $B_1$ and $B_2$ such that $B_1=\{\boldsymbol{c}\in B \vert  c_{m^\star}\leq \lfloor (v_{m^\star}+w_{m^\star})/2 \rfloor\}$ and $B_2=\{\boldsymbol{c}\in B \vert  c_{m^\star}\geq \lceil (v_{m^\star}+w_{m^\star})/2 \rceil\}$, where $m^\star = \argmax_{m\in\mathcal{M}_\text{max}}(w_{m}-v_{m})$.
    \item {\tt Reduce}: A box $B\hspace{-0.5mm}=\hspace{-0.5mm}[\boldsymbol{v},\boldsymbol{w}]$ is reduced to $B'\hspace{-0.5mm}=\hspace{-0.5mm}[\boldsymbol{v}',\boldsymbol{w}']$ such that $\boldsymbol{v}'\hspace{-0.5mm}=\hspace{-0.5mm}\boldsymbol{w}\hspace{-0.5mm}-\hspace{-0.5mm}\sum^{M_\text{max}}_{m=1}\alpha_m(w_m\hspace{-0.5mm}-\hspace{-0.5mm}v_m)\mathbf{e}_m$ and $\boldsymbol{w}'\hspace{-0.5mm}=$ $\boldsymbol{v}'\hspace{-0.5mm}+\hspace{-0.5mm}\sum^{M_\text{max}}_{m=1}\beta_m(w_m\hspace{-0.5mm}-\hspace{-0.5mm}v'_m)\mathbf{e}_m$, where $\alpha_m\hspace{-0.5mm}=\hspace{-0.5mm}\sup\{\alpha\vert \alpha\hspace{-0.5mm}\in[0,1],$ $g_1(\boldsymbol{v})-h_1(\boldsymbol{w}-\alpha(w_m-v_m)\mathbf{e}_m)\leq 0, f^+(\boldsymbol{w}-\alpha(w_m-v_m)\mathbf{e}_m)-f^-(\boldsymbol{v})\geq \nu(B)\}$ and $\beta_m\hspace{-0.5mm}=\hspace{-0.5mm}\sup\{\beta\vert \beta\hspace{-0.5mm}\in[0,1],$ $g_1(\boldsymbol{v}'-\beta(w_m-v'_m)\mathbf{e}_m)-h_1(\boldsymbol{w})\leq 0, f^+(\boldsymbol{w})-f^-(\boldsymbol{v}'+\beta(w_m-v'_m)\mathbf{e}_m)\hspace{-0.5mm}\geq\hspace{-0.5mm}\nu(B)\}$.
    Here $\nu(B)=\max_{\boldsymbol{c}\in B}$ $\left(f^+(\boldsymbol{c}) - f^-(\boldsymbol{c})\right)$, and $\mathbf{e}_m$ is the $m$-th column of $\mathbf{I}_{M_\text{max}}$.
\end{itemize}
The overall procedure is summarized in Algorithm~\ref{alg:DMO}.
Starting from the box $[\boldsymbol{0},\boldsymbol{1}]$, smaller boxes are generated using the \emph{Branch} (i.e., cutting the box in half) and \emph{Reduce} (i.e., cutting down the edges of each box) steps.
Then, by keeping the boxes that satisfy the boundary condition (e.g., $f^+(\boldsymbol{w}')-f^-(\boldsymbol{v}') < \nu^\star$ in Algorithm~\ref{alg:DMO}) and discarding the rest, the solution range is narrowed down by the algorithm.
This step is repeated until (i) there is no more box satisfying the boundary condition or (ii) a box whose $f^+(\boldsymbol{w}')-f^-(\boldsymbol{v}')$ is $\delta$-accurate to the bound.

\begin{algorithm}[h]
{\small
\caption{Robust Sensor Selection based on DMO}
\label{alg:DMO}
\begin{algorithmic}
\REQUIRE $\mu$, $\delta$, $\epsilon_m$ and $\boldsymbol{u}_m$, $\forall m \in \mathcal{M}_\text{max}$
\STATE Initialize $i=1$, $B=[\boldsymbol{0},\boldsymbol{1}]$, $\mathcal{B}_i=\{B\}$, $\mathcal{R}_i=\varnothing$, and $\nu^\star=-\infty$
\WHILE{TRUE}
\STATE Reduce each box $B\in\mathcal{B}_i$ into $B'$ and $\mathcal{R}_i = \mathcal{R}_i \cup B'$
\STATE Find $\boldsymbol{c}^{(i)}=\argmax_{\hspace{-3mm}\boldsymbol{c}=\left\lceil \frac{(\boldsymbol{v}'+\boldsymbol{w}')}{2} \right\rceil,\forall B'\in \mathcal{R}_i} (f^+(\boldsymbol{c})-f^-(\boldsymbol{c})) > \nu^\star$
\IF{$\boldsymbol{c}^{(i)}$ exists}
    \STATE $\nu^\star=f^+(\boldsymbol{c}^{(i)})-f^-(\boldsymbol{c}^{(i)})$
\ELSE
    \STATE $\boldsymbol{c}^{(i)} = \boldsymbol{c}^{(i-1)}$
\ENDIF
\STATE Delete every $B'\in\mathcal{R}_i$ satisfying $f^+(\boldsymbol{w}')-f^-(\boldsymbol{v}') < \nu^\star$
\STATE $\mathcal{R}_{i+1} = \mathcal{R}_i$
\IF{$\mathcal{R}_{i+1}=\varnothing$}
    \RETURN $\boldsymbol{c}_\text{R3}^\star=\boldsymbol{c}^{(i)}$
\ELSE
    \STATE $B^{(i)}=\argmax_{B'\in\mathcal{R}_{i+1}}f^+(\boldsymbol{w}')-f^-(\boldsymbol{v}')$
    \IF{$\nu^\star \geq \delta (f^+(\boldsymbol{w}^{(i)})-f^-(\boldsymbol{v}^{(i)}))$}
        \RETURN $\boldsymbol{c}_\text{R3}^\star=\boldsymbol{c}^{(i)}$
    \ELSE
        \STATE Branch $B^{(i)}$ into $B^{(i)}_1$ and $B^{(i)}_2$
        \STATE $\mathcal{R}_{i+1}=\mathcal{R}_{i+1}\backslash B^{(i)}$ and $\mathcal{B}_{i+1}=\{B^{(i)}_1,B^{(i)}_2\}$
    \ENDIF
\ENDIF
\STATE $i=i+1$
\ENDWHILE
\end{algorithmic}
}
\end{algorithm}

\section{Numerical Simulations}\label{sec:simulations}

\subsection{Simulation Setup and Metrics}\label{ssec:simulations_setup}

\begin{figure}[!t]
    \vspace{-4mm}
    \centering
    \includegraphics[width=0.95\linewidth]{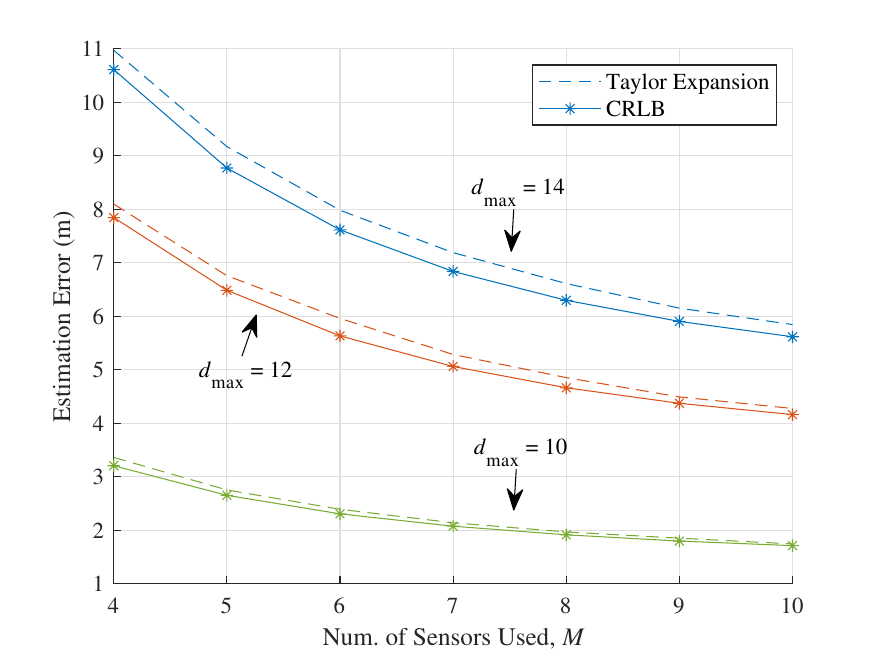}
    \caption{Comparison between the CRLB and a MLE method (Taylor expansion) over different values of $M$.}
    \vspace{-4mm}
    \label{fig:MLE_and_corr}
\end{figure}

Unless stated otherwise, we use $M_\text{max}=14$ sensors placed inside the sensor space confined by $d_\text{s}=4$m with target space covering $d_\text{max} = 14$m.
The location of sensor $m$, i.e., $\boldsymbol{\ell}_m^\mathsf{s}$, $\forall m\in\mathcal{M}_\text{max}$, is randomly generated for each experiment.
According to conventional signal modeling, TOA and RSS measurement noises for sensor $m$ are assumed to have $\sigma^2_{\text{T},m} = \frac{c^2}{8\pi \text{SNR}_mW^2}$~\cite{Huang14} with $\text{SNR}_m=d_m^{-\xi}$ and $\sigma^2_{\text{R},m} = \big(\frac{\log 10}{10\xi}\big)^2\sigma^2_{\text{S},m}$~\cite{Catovic04}, where $W$ is the signal bandwidth and $\sigma^2_{\text{S},m}$ is the shadowing variance~\cite{Laaraiedh12}.
We set $\xi = 2$, $W = 500 \text{ MHz}$, and $\sigma^2_{\text{S},m}=0.83$.
To obtain CRLB plots, an average was taken across 10 different experiments, each of which had either 152 randomly generated $\boldsymbol{\ell}$ for dynamic sensor selection or 152 evenly distributed $\boldsymbol{\ell}^\mathsf{t}_g$ for robust sensor selection inside the target space.
We measure the complexity of algorithms in terms of their runtime computed in the implementation software MATLAB R2021a. The results present the average of 10,000 independent runs/executions of each algorithm.

We employ the CRLB expression $\sigma_{\mathcal{M}}^2(\boldsymbol{\ell})$ in~\eqref{eq:CRLB_hybrid_corr_1} as our positioning accuracy metric for sensor selection.
To demonstrate its appropriateness, we implement $f_\text{est}$, which utilizes the Taylor expansion~\cite{Foy76}, and compare its positioning performance to the CRLB.
The result across different values of $M$ and $d_\text{max}$ is shown in Fig.~\ref{fig:MLE_and_corr}.
To generate the plot, the solution to $\boldsymbol{\mathcal{P}}_\text{D}$ was found via exhaustive search for each value of $M$ and used to compute both $\text{MSE}_\mathcal{M}(\boldsymbol{\ell})$ of $f_\text{est}$ (dashed) and $\sigma^2_\mathcal{M}(\boldsymbol{\ell})$ (solid).
As seen in Fig.~\ref{fig:MLE_and_corr}, for all values of $M$ and $d_\text{max}$, the CRLB is just slightly lower than the MSE performance of $f_\text{est}$. 
We thus conclude that the CRLB is a valid metric to quantify the selection-dependent accuracy performance of wireless positioning.

\vspace{-3mm}
\subsection{Dynamic Sensor Selection}\label{ssec:simulations_dynamic}

We evaluate the performance of our dynamic sensor selection algorithms from Section~\ref{sec:dynamic}.
We refer to Algorithms~\ref{alg:dynamic_trace} and~\ref{alg:dynamic_fraction} as greedy sensor selection using T-CRLB (GSS-T) and greedy sensor selection using F-CRLB (GSS-F), respectively.

\begin{figure}[!t]
    \centering
    \includegraphics[width=0.95\linewidth]{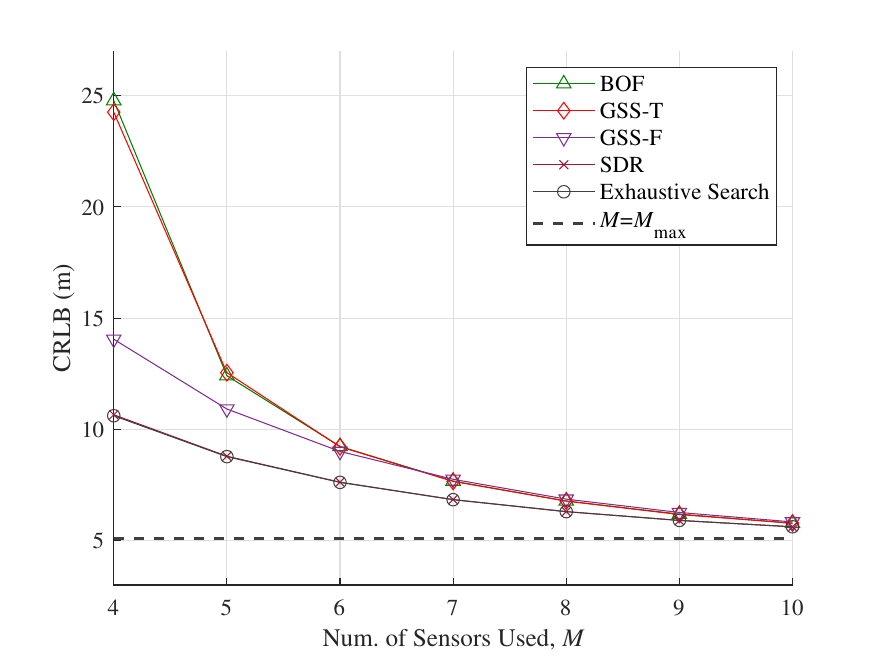}
    \caption{Comparison of the CRLB obtained from different dynamic sensor selection algorithms over $M$ when $M_\text{max}=14$.}
    \vspace{-4mm}
    \label{fig:dynamic_accuracy}
\end{figure}

We first consider the average CRLBs obtained by our algorithms and three benchmarks: (i) exhaustive search, (ii) SDR~\cite{Dai20}, and (iii) BOF greedy selection~\cite{Zhao19}.
According to the results shown in Fig.~\ref{fig:dynamic_accuracy}, we see that exhaustive search and SDR provide the best performance in CRLB minimization.
Note that greedy sensor selection algorithms often converge to a suboptimal solution of problem $\mathcal{P}_\text{D}$ since the successive decisions made for each sensor selection may not always lead to a globally optimal decision.
Out of the greedy selection algorithms, GSS-F shows the most comparable performance to the optimal case.
We also see that GSS-T has its performance equivalent to BOF, which implies that our proposed sensor selection metric matches the efficacy of completely evaluating the CRLB expression.
Additionally, we see that both GSS-T and BOF have inferior performance compared to GSS-F.
This is since more sensors must be heuristically selected for GSS-T and BOF, emphasizing the importance of early stages in greedy selection.
For each algorithm, the improvement in the CRLB diminishes upon increasing $M$, and the performance becomes close to the case of $M=M_\text{max}$ (dashed line) when $M$ approaches~$10$.
Moreover, the performance gap among the algorithms becomes almost negligible at $M=10$.
This implies a stable positioning performance (i.e., a consistent CRLB regardless of the algorithm) can be achieved once the number of used/selected sensors in the sensor space is large enough.
It is worth mentioning that these numbers are application-specific and depend on the sensor space volume and the region occupied by the target candidate locations.

Overall, comparing GSS-T and GSS-F to the benchmarks verifies that our proposed selection strategies only marginally degrade positioning accuracy performance compared to the computationally intensive SDR and exhaustive search.
We move now to quantify the improvement in complexity that they provide.
In Fig.~\ref{fig:dynamic_complexity_M}, the running times of different sensor selection algorithms are shown.
We see that GSS-F takes the shortest time to complete the sensor selection compared to other greedy selection algorithms.
This supports our analysis in Section~\ref{ssec:dynamic_complexity}, which expects GSS-F to be the quickest from having low leading coefficients.
We also see that the GSS-T algorithm only requires about half the running time taken by BOF.
Therefore, we can consider GSS-T and GSS-F to be superior to BOF in the sense that they provide similar or better accuracy performance in much less time.
Note that SDR, which is not included in Fig.~\ref{fig:dynamic_complexity_M}, shows the worst runtime performance with 0.56 seconds in average.
The fast execution times of our proposed algorithms allow near real-time sensor selection for moving targets as well.
In particular, their fast sensor selection times allow the target to only move a short distance before the actual positioning step is carried by the system.

\begin{figure}[!t]
    \centering
    \includegraphics[width=0.9\linewidth]{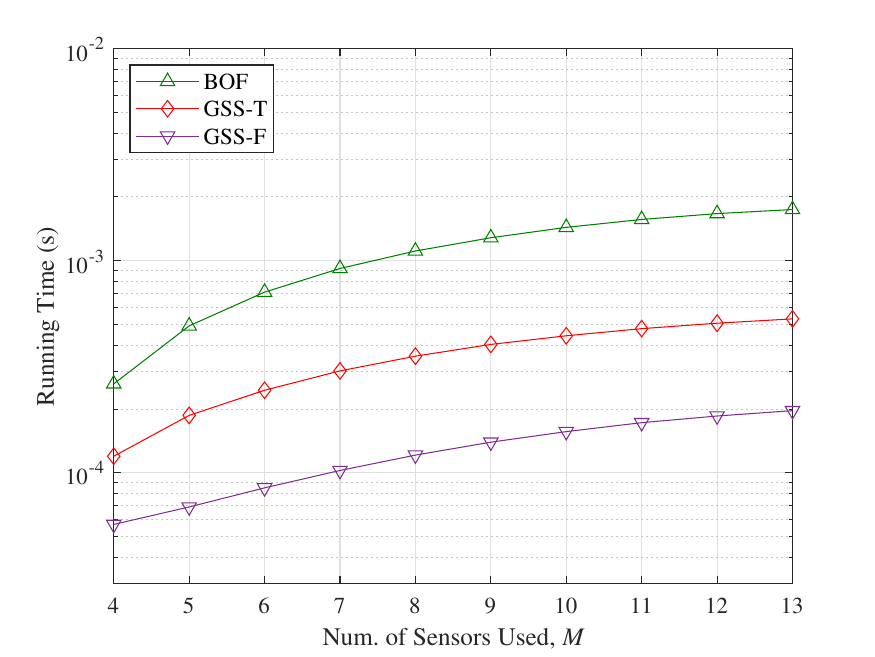}
    \caption{Comparison between running time from different dynamic sensor selection algorithms over $M$ when $M_\text{max}=14$.}
    \label{fig:dynamic_complexity_M}
\end{figure}

\begin{figure}[!t]
    \centering
    \includegraphics[width=0.95\linewidth]{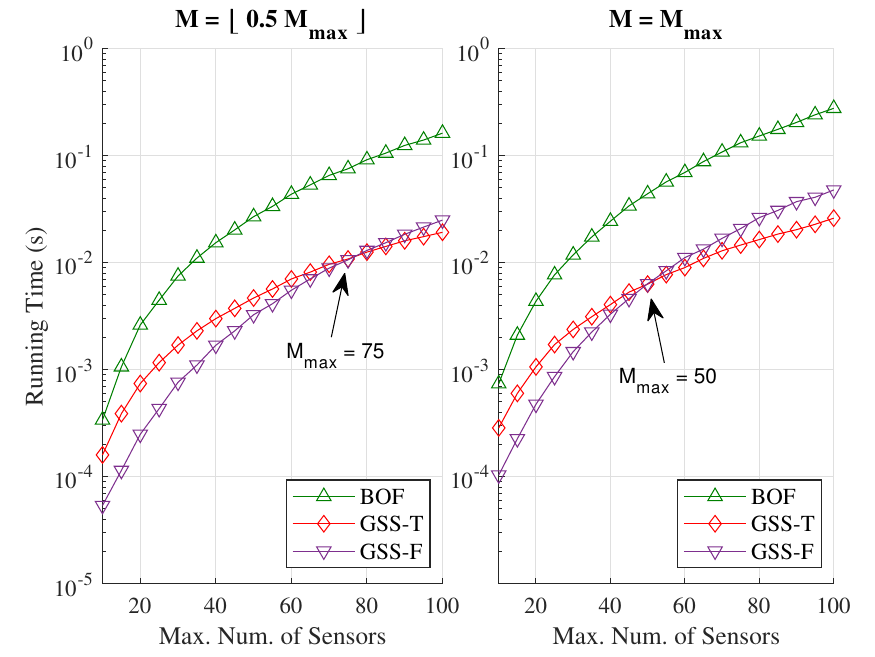}
    \caption{Comparison of the running time of different dynamic sensor selection algorithms over $M_\text{max}$ when $M\hspace{-0.5mm}=\hspace{-0.5mm}\lfloor0.5M_\text{max}\rfloor$ (left) and $M\hspace{-0.5mm}=\hspace{-0.5mm}M_\text{max}$ (right).}
    \vspace{-4mm}
    \label{fig:dynamic_complexity_M_max}
\end{figure}

Next, we evaluate the asymptotic complexity results from Section~\ref{ssec:dynamic_complexity} by considering large numbers for $M_\text{max}$.
In Fig.~\ref{fig:dynamic_complexity_M_max}, two plots comparing the average running times of BOF, GSS-T, and GSS-F over different values of $M_\text{max}$ for $M=\lfloor0.5M_\text{max}\rfloor$ (left) and $M=M_\text{max}$ (right) are shown. 
We see that both of our algorithms take less time than BOF as $M_\text{max}$ increases.
As we discussed in Section~\ref{ssec:dynamic_complexity}, GSS-T and GSS-F have complexities $\mathcal{O}(M_\text{max}^2)$ and $\mathcal{O}(M_\text{max}^3)$, respectively, in terms of the number of arithmetic operations involved.
Since GSS-F has less leading coefficients than GSS-T, we see that GSS-F has faster running time for smaller $M_\text{max}$ but eventually surpasses GSS-T, when $M_\text{max}\geq75$ for $M=\lfloor0.5M_\text{max}\rfloor$ and $M_\text{max}\geq50$ for $M=M_\text{max}$.
Note that the BOF algorithm of complexity $\mathcal{O}(M_\text{max}^3)$ having the same asymptotic behavior as GSS-F is verified by their slope being similar for large $M_\text{max}$.

\begin{figure}[!t]
    \centering
    \includegraphics[width=0.95\linewidth]{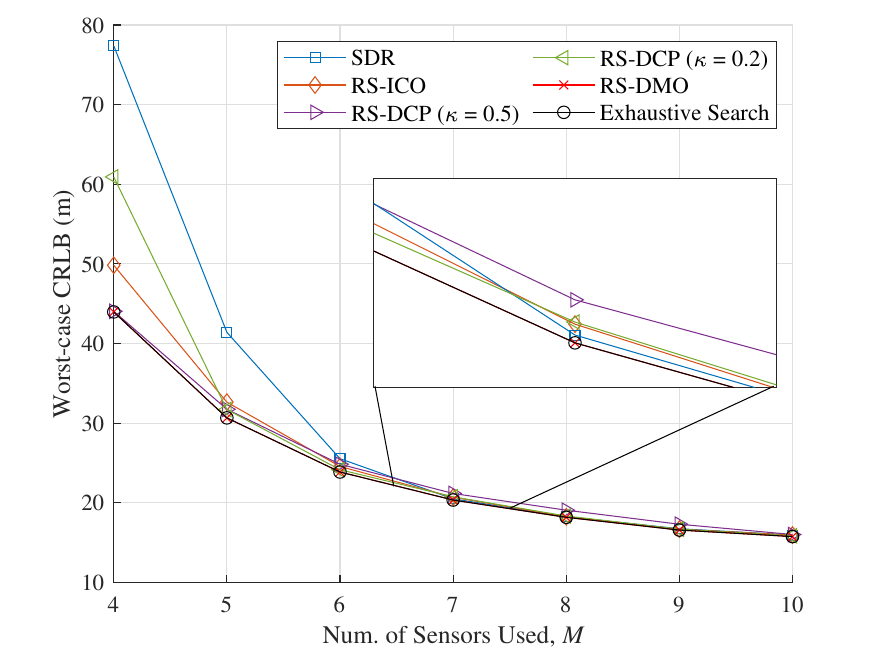}
    \caption{Comparison in the CRLB performance among multiple robust sensor selection algorithms over different values of $M$.}
    \label{fig:robust_accuracy}
\end{figure}

\vspace{-3mm}
\subsection{Robust Sensor Selection} \label{ssec:simulations_robust}

We evaluate the performance of our robust sensor selection algorithms.
We refer to Algorithms~\ref{alg:iter-CVX},~\ref{alg:DC}, and~\ref{alg:DMO} as robust sensor selections using ICO (RSS-ICO), DCP (RSS-DCP), and DMO (RSS-DMO), respectively.
We set $\varepsilon=0.05$ for RSS-DCP and $\mu=100$, $\delta=0.05$ for RSS-DMO.
A plot comparing the average worst-case CRLBs of our selection algorithms and the benchmarks: (i) exhaustive search and (ii) SDR, is shown in Fig.~\ref{fig:robust_accuracy}.
As discussed in Section~\ref{ssec:robust_algorithm3}, RSS-DMO does not rely on convex relaxation and provides optimal performance (i.e., same performance as exhaustive search) in the worst-case CRLB minimization.
Despite being suboptimal, both RSS-ICO and RSS-DCP show significant improvement compared to SDR.

To evaluate the impact of penalty term on RSS-DCP, results with different $\kappa$ values are included in Fig.~\ref{fig:robust_accuracy}.
For small $M$, greater penalties are needed to force $\|\boldsymbol{c}^\star_\text{R2}\|_0=M$ since $\|\boldsymbol{c}^\star_\text{R}\|_0$ is typically far greater than $M$.
This is verified by our result in which RSS-DCP shows near-optimal performance with $\kappa=0.5$ but degraded performance with $\kappa=0.2$ for small $M$.
Note that the opposite behavior is observed for larger $M$, i.e., RSS-DCP with $\kappa=0.2$ shows better performance.
This is since $\|\boldsymbol{c}^\star_\text{R}\|_0$ is already close to $M$, and too much penalty prevents RSS-DCP from finding more robust solutions.

\begin{table}[t]
\captionsetup{justification=centering, labelsep=newline}
\centering
\caption{Runtime measurements in seconds of robust sensor selection algorithms over different values of $M$, $M_\text{max}$, and $G$.}
\label{tb:robust_runtime_1}
\begin{tabular}{|p{1.25cm}||p{0.46cm}|p{0.46cm}|p{0.46cm}|p{0.46cm}|p{0.46cm}|p{0.46cm}|p{0.46cm}|p{0.46cm}|} 
    \hline
    & \multicolumn{4}{|c|}{$M_\text{max}=10$} & \multicolumn{4}{|c|}{$M_\text{max}=15$} \tabularnewline
    \hline
    & \multicolumn{2}{|c|}{$G=20$} & \multicolumn{2}{|c|}{$G=40$} &  \multicolumn{2}{|c|}{$G=20$} & \multicolumn{2}{|c|}{$G=40$} \tabularnewline
    \hline
    \centering $M$ & \centering $4$ & \centering $6$ & \centering $4$ & \centering $6$ & \centering $4$ & \centering $6$ & \centering $4$ & \centering $6$ \tabularnewline
    \hline
    RSS-ICO & \centering 3.59 & \centering 6.21 & \centering 6.65 & \centering 10.21 & \centering 5.00 & \centering 7.46 & \centering 7.27 & \centering 11.19 \tabularnewline
    \hline
    RSS-DCP & \centering 4.96 & \centering 4.16 & \centering 8.23 & 7.30 & 5.87 & 5.44 & 9.21 & 7.35 \tabularnewline
    \hline
    RSS-DMO & \centering 2.38 & \centering 1.34 & 2.81 & 1.49 & \centering 19.82 & \centering 14.65 & \centering 27.92 & \centering 21.72 \tabularnewline
    \hline
\end{tabular}
\vspace{-4mm}
\end{table}

We now evaluate the runtime performance of our robust sensor selection algorithms.
To focus on comparative analysis, we measure runtimes with $M$, $G$, and $M_\text{max}$ varying over two different values.
The results are shown in Table~\ref{tb:robust_runtime_1}.
Both RSS-ICO and RSS-DCP provide the runtimes strictly proportional to $G$ as the algorithms conduct convex optimization over at least $G$ distinct constraints.
For RSS-ICO, since $M$ iterations are required to complete its selection, we see that the runtime also increases with $M$.
In contrast, the runtime of RSS-DCP decreases for increased $M$ because the algorithm is likely to find converged solutions quicker over the feasible space defined by a larger $M$.
Regarding our optimal algorithm RSS-DMO, we see a significant increase in the runtime for an increased $M_\text{max}$.
As RSS-DMO adopts the BRB technique to find solutions, $M_\text{max}$ strictly determines the dimension of the boxes, which directly impacts the algorithm complexity.

Different runtime behaviors shown by our robust sensor selection algorithms indicate that no single algorithm claims both computation and accuracy advantages over the others, which lead to a pareto solution in the complexity/accuracy space.
For example, RSS-DMO guarantees the optimal performance in the worst-case CRLB minimization, but is not always the best option if we must consider computational complexity.
Thus, tradeoffs between the accuracy and complexity performance should be well considered regarding our algorithms.

\vspace{-2mm}
\section{Conclusion}\label{sec:conclusion}

We have considered both dynamic and robust sensor selection problems in 3D wireless positioning with TOA/RSS hybrid measurements.
After formulating the optimization problems using the CRLB as a performance metric, trace and fractional forms of CRLB were derived and used for developing sensor selection strategies.
To address the dynamic sensor selection, two greedy selection algorithms were proposed, one based on each CRLB form, and shown to achieve substantial reductions in computational complexity, both in theory and experimentally, for comparable positioning accuracy.
Three different strategies were developed for the robust sensor selection, each having different tradeoffs between complexity and optimality guarantee in minimizing the worst-case CRLB.
Developing a joint dynamic and robust sensor selection strategy that is adaptable to various system conditions is left as our potential future work.

\appendices

\vspace{-2mm}
\section{Steps for deriving~\eqref{eq:FIM_hybrid_entry}}\label{appendix:FIM}

Using~\eqref{eq:FIM_hybrid_1} and~\eqref{eq:log-likelihood_hybrid}, the element of $\mathcal{I}_\mathcal{M}(\boldsymbol{\ell})$ is evaluated as
\vspace{-2mm}
\begin{align}
    \mathcal{I}_{\mathcal{M}}^{(vw)} &= -\mathbb{E}\left[\partial/\partial w\left\{\partial/\partial v\; l_{\mathcal{M}}(\widehat{\boldsymbol{q}}_{\mathcal{M}}\vert \boldsymbol{\ell})\right\}\right] \label{eq:FIM_hybrid_entry1} \\
    &\hspace{-9mm}= \mathbb{E}\hspace{-0.5mm}\left[\partial/\partial w\left(\partial(\widehat{\boldsymbol{q}}_{\mathcal{M}}-\boldsymbol{q}_{\mathcal{M}})\hspace{-0.5mm}^\top/\partial v\:\mathbf{R}_{\mathcal{M}}^{-1}(\widehat{\boldsymbol{q}}_{\mathcal{M}}-\boldsymbol{q}_{\mathcal{M}})\right)\right] \label{eq:FIM_hybrid_entry2} \\
    &\hspace{-9mm}=\mathbb{E}\big[\partial(\widehat{\boldsymbol{q}}_{\mathcal{M}}-\boldsymbol{q}_{\mathcal{M}})\hspace{-0.5mm}^\top\hspace{-1mm}/\partial v\:\mathbf{R}_{\mathcal{M}}^{-1}\partial(\widehat{\boldsymbol{q}}_{\mathcal{M}}-\boldsymbol{q}_{\mathcal{M}})/\partial w \nonumber \\
    &\quad\quad\;\quad\quad+\partial^2(\widehat{\boldsymbol{q}}_{\mathcal{M}}-\boldsymbol{q}_{\mathcal{M}})\hspace{-0.5mm}^\top\hspace{-1mm}/\partial v \partial w\:\mathbf{R}_{\mathcal{M}}^{-1}(\widehat{\boldsymbol{q}}_{\mathcal{M}}\hspace{-0.5mm}-\hspace{-0.5mm}\boldsymbol{q}_{\mathcal{M}})\big] \nonumber \\ 
    &\hspace{-9mm}=\mathbb{E}\big[\partial\boldsymbol{q}_{\mathcal{M}}^\top/\partial v\mathbf{R}_{\mathcal{M}}^{-1}\partial\boldsymbol{q}_{\mathcal{M}}/\partial w\hspace{-0.5mm}-\hspace{-0.5mm}\partial^2\boldsymbol{q}_{\mathcal{M}}^\top/\partial v \partial w\mathbf{R}_{\mathcal{M}}^{-1}\hspace{-0.4mm}(\widehat{\boldsymbol{q}}_{\mathcal{M}}\hspace{-0.5mm}-\hspace{-0.5mm}\boldsymbol{q}_{\mathcal{M}})\big]\hspace{-1mm} \label{eq:FIM_hybrid_entry4} \\
    &\hspace{-9mm}=\partial\boldsymbol{q}_{\mathcal{M}}^\top/\partial v\:\mathbf{R}_{\mathcal{M}}^{-1}\partial\boldsymbol{q}_{\mathcal{M}}/\partial w, \label{eq:FIM_hybrid_entry5}
    \vspace{-2mm}
\end{align}
for $v,w\in\{x,y,z\}$.
Equality in~\eqref{eq:FIM_hybrid_entry4} holds because the derivatives of $\widehat{\boldsymbol{q}}_{\mathcal{M}}$ are zero.
Note that, with MLE, $\widehat{\boldsymbol{q}}_{\mathcal{M}}$ is simply treated as an observation vector.
The last equality is true as the expectation only applies to $\widehat{\boldsymbol{q}}_{\mathcal{M}}$, and $\mathbb{E}[\widehat{\boldsymbol{q}}_{\mathcal{M}}]=\boldsymbol{q}_{\mathcal{M}}$ makes the last term in~\eqref{eq:FIM_hybrid_entry4} zero.

\vspace{-3mm}
\section{Proof of Proposition~\ref{prop:CRLB}}\label{appendix:CRLB}
\vspace{-1mm}

We begin the proof by considering~\eqref{eq:CRLB_hybrid_corr_1}.
Let us define $\mathbf{U}_{\mathcal{M}}=\left[\{\sqrt{\epsilon_{m}}\boldsymbol{u}_{m}\}_{m\in\mathcal{M}}\right]$ to be a $3\times M$ matrix having $M$ columns of $\sqrt{\epsilon_{m}}\boldsymbol{u}_{m}$ for $m\in\mathcal{M}$.
Then, \eqref{eq:CRLB_hybrid_corr_1} can be rewritten as
\begin{equation}
    \sigma^2_{\mathcal{M}}(\boldsymbol{\ell})=\textrm{tr}\left\{\big(\mathbf{U}_{\mathcal{M}}\mathbf{U}_{\mathcal{M}}^\top\big)^{-1}\right\}=\frac{\textrm{tr}\{\textrm{adj}\big(\mathbf{U}_{\mathcal{M}}\mathbf{U}_{\mathcal{M}}^\top\big)\}}{\det(\mathbf{U}_{\mathcal{M}}\mathbf{U}_{\mathcal{M}}^\top)}
    \label{eq:CRLB_frac}.
\end{equation}
The denominator of~\eqref{eq:CRLB_frac} can be derived as
\begin{align}
    &\hspace{-2mm}\det(\mathbf{U}_{\mathcal{M}}\mathbf{U}_{\mathcal{M}}^\top) \nonumber \\
    &\hspace{-2mm}=\hspace{-2.5mm}\sum_{m_1\in\mathcal{M}}\hspace{-0.5mm}\sum_{\substack{m_2\in\mathcal{M}\\m_2>m_1}}\hspace{-0.5mm}\sum_{\substack{m_3\in\mathcal{M}\\m_3>m_2}} \hspace{-1.5mm}\det([\epsilon_{m_1}^{\nicefrac{1}{2}}\boldsymbol{u}_{m_1},\epsilon_{m_2}^{\nicefrac{1}{2}}\boldsymbol{u}_{m_2},\epsilon_{m_3}^{\nicefrac{1}{2}}\boldsymbol{u}_{m_3}])^2 \nonumber \\
    &\hspace{-2mm}=\hspace{-2.5mm}\sum_{m_1\in\mathcal{M}}\hspace{-0.5mm}\sum_{\substack{m_2\in\mathcal{M}\\m_2>m_1}}\hspace{-0.5mm}\sum_{\substack{m_3\in\mathcal{M}\\m_3>m_2}} \hspace{-1mm}\epsilon_{m_1}\epsilon_{m_2}\epsilon_{m_3}\big((\boldsymbol{u}_{m_1}\hspace{-1mm}\times\boldsymbol{u}_{m_2})\cdot\boldsymbol{u}_{m_3}\big)^2 \nonumber  \\
    &\hspace{-2mm}=\hspace{-2.5mm}\sum_{m_1\in\mathcal{M}}\hspace{-0.5mm}\sum_{\substack{m_2\in\mathcal{M}\\m_2>m_1}}\hspace{-0.5mm}\sum_{\substack{m_3\in\mathcal{M}\\m_3>m_2}} \epsilon_{m_1}\epsilon_{m_2}\epsilon_{m_3} \nonumber \\
    &\quad\quad\times\hspace{-0.5mm}\left(\|\boldsymbol{u}_{m_1}\hspace{-1mm}\times\boldsymbol{u}_{m_2}\|\|\boldsymbol{u}_{m_3}\|\cos\big(\nicefrac{\pi}{2}-\phi_{m_1m_2m_3}\big)\hspace{-0.5mm}\right)^2 \label{eq:CRLB_derived_bot_proof_3} \\
    &\hspace{-2mm}=\hspace{-2.5mm}\sum_{m_1\in\mathcal{M}}\hspace{-0.5mm}\sum_{\substack{m_2\in\mathcal{M}\\m_2>m_1}}\hspace{-0.5mm}\sum_{\substack{m_3\in\mathcal{M}\\m_3>m_2}} \epsilon_{m_1}\epsilon_{m_2}\epsilon_{m_3} \nonumber \\
    &\quad\times\hspace{-0.5mm}\big(\|\boldsymbol{u}_{m_1}\|\|\boldsymbol{u}_{m_2}\|\sin\theta_{m_1m_2}\|\boldsymbol{u}_{m_3}\|\sin{\phi_{m_1m_2m_3}}\big)^2 \label{eq:CRLB_derived_bot_proof_4} \\
    &\hspace{-2mm}=\hspace{-2.5mm}\sum_{m_1\in\mathcal{M}}\hspace{-0.5mm}\sum_{\substack{m_2\in\mathcal{M}\\m_2>m_1}}\hspace{-0.5mm}\sum_{\substack{m_3\in\mathcal{M}\\m_3>m_2}} \hspace{-1.5mm} \epsilon_{m_1}\epsilon_{m_2}\epsilon_{m_3}\sin^2{\theta_{m_1m_2}}\sin^2{\phi_{m_1m_2m_3}}, \label{eq:CRLB_derived_bot_proof_5}
\end{align}
where the operations $\times$ and $\cdot$ between two vectors indicate cross-product and inner-product, respectively. The first equality is from the Cauchy-Binet formula and property that $\vert A\vert \vert A^\top\vert \hspace{-1mm}=\hspace{-1mm}\vert A\vert ^2$ for any square matrix $A$.
The second equality holds as the determinant of a matrix is equal to the volume created by column vectors.
Equalities in~\eqref{eq:CRLB_derived_bot_proof_3} and~\eqref{eq:CRLB_derived_bot_proof_4} hold due to $\boldsymbol{u}_a\cdot\boldsymbol{u}_b=\|\boldsymbol{u}_a\|\|\boldsymbol{u}_b\|\cos\theta_{ab}$ and $\|\boldsymbol{u}_a\times\boldsymbol{u}_b\|=\|\boldsymbol{u}_a\|\|\boldsymbol{u}_b\|\sin\theta_{ab}$, respectively.
The last equality holds since $\|\boldsymbol{u}_m\|=1$, $\forall m$.

The numerator of~\eqref{eq:CRLB_frac} is derived as follows.
Since we are only interested in the diagonal terms, each of which is equal to the determinant of $2\times2$ matrix obtained after removing the corresponding row and column,
\begin{align}
    &\hspace{-1mm}\textrm{tr}\{\text{adj}\big(\mathbf{U}_{\mathcal{M}}\mathbf{U}_{\mathcal{M}}^\top\big)\} \nonumber \\
    &\hspace{-1mm}=\hspace{-2mm}\sum_{m_1\in\mathcal{M}}\sum_{\substack{m_2\in\mathcal{M}\\m_2>m_1}}
    \epsilon_{m_1}\epsilon_{m_2}\Bigg(\begin{vmatrix}\nicefrac{(x_{m_1}^\mathsf{s}-x)}{d_{m_1}} \hspace{-1.5mm}&\hspace{-1.5mm} \nicefrac{(x_{m_2}^\mathsf{s}-x)}{d_{m_2}} \\ \nicefrac{(y_{m_1}^\mathsf{s}-y)}{d_{m_1}} \hspace{-1.5mm}&\hspace{-1.5mm} \nicefrac{(y_{m_2}^\mathsf{s}-y)}{d_{m_2}}\end{vmatrix}^2 \nonumber \\
    &\hspace{-1mm}+\hspace{-0.5mm}\begin{vmatrix}\nicefrac{(x_{m_1}^\mathsf{s}-x)}{d_{m_1}} \hspace{-1.5mm}&\hspace{-1.5mm} \nicefrac{(x_{m_2}^\mathsf{s}-x)}{d_{m_2}} \\ \nicefrac{(z_{m_1}^\mathsf{s}-z)}{d_{m_1}} \hspace{-1.5mm}&\hspace{-1.5mm} \nicefrac{(z_{m_2}^\mathsf{s}-z)}{d_{m_2}}\end{vmatrix}^2\hspace{-2mm}+\hspace{-0.5mm}\begin{vmatrix}\nicefrac{(y_{m_1}^\mathsf{s}-y)}{d_{m_1}} \hspace{-1.5mm}&\hspace{-1.5mm} \nicefrac{(y_{m_2}^\mathsf{s}-y)}{d_{m_2}} \\ \nicefrac{(z_{m_1}^\mathsf{s}-z)}{d_{m_1}} \hspace{-1.5mm}&\hspace{-1.5mm} \nicefrac{(z_{m_2}^\mathsf{s}-z)}{d_{m_2}}\end{vmatrix}^2\hspace{-0.5mm}\Bigg) \nonumber  \\
    &\hspace{-1mm}=\hspace{-2mm}\sum_{m_1\in\mathcal{M}}\sum_{\substack{m_2\in\mathcal{M}\\m_2>m_1}} \hspace{-2mm}\epsilon_{m_1}\epsilon_{m_2}\big(\|\boldsymbol{u}_{m_1}\|^2\|\boldsymbol{u}_{m_2}\|^2\hspace{-1mm}-\hspace{-0.5mm}(\boldsymbol{u}_{m_1}\hspace{-1mm}\cdot\boldsymbol{u}_{m_2})^2\big) \label{eq:CRLB_derived_top_proof_2} \\
    &\hspace{-1mm}=\hspace{-2mm}\sum_{m_1\in\mathcal{M}}\sum_{\substack{m_2\in\mathcal{M}\\m_2>m_1}} \hspace{-2mm}\epsilon_{m_1}\epsilon_{m_2}\hspace{-1mm}\left(1\hspace{-1mm}-\hspace{-0.5mm}\|\boldsymbol{u}_{m_1}\|^2\|\boldsymbol{u}_{m_2}\|^2\cos^2\theta_{m_1m_2}\right) \label{eq:CRLB_derived_top_proof_3} \\
    &\hspace{-1mm}=\hspace{-2mm}\sum_{m_1\in\mathcal{M}}\sum_{\substack{m_2\in\mathcal{M}\\m_2>m_1}} \epsilon_{m_1}\epsilon_{m_2}\sin^2\theta_{m_1m_2}. \label{eq:CRLB_derived_top_proof_4}
\end{align}
The first equality holds from applying the Cauchy-Binet formula to each diagonal term. Equality in~\eqref{eq:CRLB_derived_top_proof_3} is due to $\boldsymbol{u}_a\cdot\boldsymbol{u}_b=\|\boldsymbol{u}_a\|\|\boldsymbol{u}_b\|\cos\theta_{ab}$ and $\|\boldsymbol{u}_m\|=1$, $\forall m$. The last equality holds from $\cos^2\theta+\sin^2\theta = 1$.
Rewriting~\eqref{eq:CRLB_frac} using~\eqref{eq:CRLB_derived_bot_proof_5} and~\eqref{eq:CRLB_derived_top_proof_4} completes the derivation of the F-CRLB.

\vspace{-1.5mm}
\section{Proof of Proposition~\ref{prop:SVR}}\label{appendix:SVR}

We begin the proof by considering~\eqref{eq:CRLB_frac}. Letting $\lambda_{\mathcal{M},1}$, $\lambda_{\mathcal{M},2}$, and $\lambda_{\mathcal{M},3}$ be the eigenvalues of $\mathbf{U}_{\mathcal{M}}\mathbf{U}_{\mathcal{M}}^\top$, we can express the determinant as 
\begin{equation}
    \det(\mathbf{U}_{\mathcal{M}}\mathbf{U}_{\mathcal{M}}^\top) = \prod\nolimits_{d=1}^{3}\lambda_{\mathcal{M},d}.
    \label{eq:CRLB_det}
\end{equation}
Using the property $\textrm{tr}\big\{\big(\mathbf{U}_{\mathcal{M}}\mathbf{U}_{\mathcal{M}}^\top\big)^{-1}\big\}=\sum^3_{d=1}\nicefrac{1}{\lambda_{\mathcal{M},d}}$, we express the trace of the adjugate matrix of $\mathbf{U}_{\mathcal{M}}\mathbf{U}_{\mathcal{M}}^\top$ as
\begin{align}
    \textrm{tr}\big\{\text{adj}\big(\mathbf{U}_{\mathcal{M}}\mathbf{U}_{\mathcal{M}}^\top\big)\big\} &= \textrm{tr}\big\{\big(\mathbf{U}_{\mathcal{M}}\mathbf{U}_{\mathcal{M}}^\top\big)^{-1}\big\}\det\hspace{-0.5mm}\big(\mathbf{U}_{\mathcal{M}}\mathbf{U}_{\mathcal{M}}^\top\big) \nonumber \\
    &=\sum\nolimits^3_{d=1}\hspace{-1mm}\frac{1}{\lambda_{\mathcal{M},d}}\big(\prod\nolimits_{j=1}^{3}\lambda_{\mathcal{M},j}\big)\hspace{-0.5mm} \nonumber \\
    &=\sum\nolimits_{d=1}^3\hspace{-1mm}\big(\prod\nolimits^3_{j \neq d}\lambda_{\mathcal{M},j}\big).
    \label{eq:CRLB_adj}
\end{align}
Using $\sum^3_{d=1}\nicefrac{1}{\lambda_{\mathcal{M},d}}$ and~\eqref{eq:CRLB_adj}, $\sigma^2_{\mathcal{M}}(\boldsymbol{\ell})$ can be expressed as
\begin{equation}
    \sigma^2_{\mathcal{M}}(\boldsymbol{\ell}) = \frac{\sum_{d=1}^3\big(\prod^3_{j \neq d}\lambda_{\mathcal{M},j}\big)}{\prod_{d=1}^{3}\lambda_{\mathcal{M},d}}.
    \label{eq:CRLB_frac_eig}
\end{equation}

Perceiving $\{\lambda_{\mathcal{M},d}\}^3_{d=1}$ as the magnitudes of the eigenvectors of $\mathbf{U}_{\mathcal{M}}\mathbf{U}_{\mathcal{M}}^\top$, we find that~\eqref{eq:CRLB_det} and~\eqref{eq:CRLB_adj} correspond to the volume and half of the surface area, respectively, of a rectangular prism whose dimension is defined by $\{\lambda_{\mathcal{M},d}\}^3_{d=1}$.
Therefore, $\sigma^2_{\mathcal{M}}(\boldsymbol{\ell})$ is half the ratio of surface are to volume.

\bibliographystyle{IEEEtran}
\bibliography{IEEEfull,mybib}

\begin{IEEEbiographynophoto}    
    {Myeung Suk Oh} (Student Member, IEEE) received the B.S. degree from Georgia Institute of Technology, USA and the M.S. degree from the Korea Advanced Institute of Science and Technology (KAIST) in 2013 and 2015, respectively. He is currently pursuing his Ph.D. degree in electrical engineering at Purdue University, IN, USA. His research interests include learning-aided channel estimations and sensor-based localization techniques in future wireless systems.
\end{IEEEbiographynophoto}

\vspace{-10mm}
\begin{IEEEbiographynophoto}
    {Seyyedali~Hosseinalipour} (Member, IEEE) received the B.S. degree in electrical engineering from Amirkabir University of Technology, Tehran, Iran, in 2015 with high honor and top-rank recognition. He then received the M.S. and Ph.D. degrees in electrical engineering from North Carolina State University, NC, USA, in 2017 and 2020, respectively. He was the recipient of the ECE Doctoral Scholar of the Year Award (2020) and ECE Distinguished Dissertation Award (2021) at North Carolina State University.  He was a postdoctoral researcher at Purdue University, IN, USA from 2020 to 2022. He is currently an assistant professor at the Department of Electrical Engineering at the University at Buffalo (SUNY). 
    
    He has served as the TPC Co-Chair of workshops related to distributed machine learning and edge computing held in conjunction with IEEE INFOCOM 2021\&2023, IEEE GLOBECOM 2021, IEEE ICC 2021, IEEE/CVF CVPR 2023, IEEE MSN 2021-2023, IEEE VTC 2023. His research interests include the analysis of modern wireless networks, synergies between machine learning methods and fog computing systems, distributed machine learning, and network optimization.
\end{IEEEbiographynophoto}

\vspace{-10mm}
\begin{IEEEbiographynophoto}
    {Taejoon~Kim} (Senior Member, IEEE) received the Ph.D. degree in electrical and computer engineering from Purdue University, West Lafayette, IN, USA. He is currently an Assistant Professor and Chair’s Council Faculty of electrical engineering and computer science at the University of Kansas (KU). Prior to joining KU, he was a Senior Researcher at the Nokia Bell Laboratories, Berkeley, CA, USA, a Postdoctoral Researcher at KTH, Stockholm, Sweden, and an Assistant Professor at the City University of Hong Kong. His research interests include 5G-and-beyond wireless systems, machine learning for communications, security, multiple-input multiple-output (MIMO), and statistical signal processing. He was an Associate Editor of the \emph{IEEE Transactions on Communications}. He holds 29 issued U.S. patents. He was a recipient of the Harry Talley Excellence in Teaching Award and Miller Professional Development Award in Research from the KU School of Engineering. Along with the coauthors, he won The IEEE Communications Society Stephen O. Rice Prize in 2016 and IEEE PIMRC 2012 Best Paper Award.
\end{IEEEbiographynophoto}

\vspace{-10mm}
\begin{IEEEbiographynophoto}
    {David~J.~Love} (Fellow, IEEE) received the B.S. (with highest honors), M.S.E., and Ph.D. degrees in electrical engineering from the University of Texas at Austin in 2000, 2002, and 2004, respectively. Since 2004, he has been with the Elmore Family School of Electrical and Computer Engineering at Purdue University, where he is now the Nick Trbovich Professor of Electrical and Computer Engineering. He served as a Senior Editor for IEEE Signal Processing Magazine, Editor for the IEEE Transactions on Communications, Associate Editor for the IEEE Transactions on Signal Processing, and guest editor for special issues of the IEEE Journal on Selected Areas in Communications and the EURASIP Journal on Wireless Communications and Networking. He was a member of the Executive Committee for the National Spectrum Consortium. He holds 32 issued U.S. patents. His research interests are in the design and analysis of broadband wireless communication systems, beyond-5G wireless systems, multiple-input multiple-output (MIMO) communications, millimeter wave wireless, software defined radios and wireless networks, coding theory, and MIMO array processing. Dr. Love is a Fellow of the American Association for the Advancement of Science (AAAS) and was named a Thomson Reuters Highly Cited Researcher (2014 and 2015). Along with his co-authors, he won best paper awards from the IEEE Communications Society (2016 Stephen O. Rice Prize and 2020 Fred W. Ellersick Prize), the IEEE Signal Processing Society (2015 IEEE Signal Processing Society Best Paper Award), and the IEEE Vehicular Technology Society (2010 Jack Neubauer Memorial Award).
\end{IEEEbiographynophoto}

\begin{IEEEbiographynophoto}
    {James~V.~Krogmeier} (Senior Member, IEEE) received the B.S.E.E. degree from the University of Colorado Boulder, Boulder, CO, USA, and the M.S. and Ph.D. degrees from the University of Illinois at Urbana-Champaign, Champaign, IL, USA. He has industry experience in telecommunications and is a Founding Member of two software startup companies. He is currently a Professor of electrical and computer engineering with Purdue University, West Lafayette, IN, USA. He has authored or coauthored many technical papers in refereed journals and conference proceedings of the IEEE, the ASABE, and the Transportation Research Board, and is a Co-Inventor of five U.S. patents. His research interests include the applications of statistical signal and image processing in agriculture, intelligent transportation systems, sensor networking, and wireless communications. His research has been funded by the USDA-NIFA, the NSF, the DARPA, the Indiana Department of Transportation, the Federal Highway Administration, and industry. He was on a number of IEEE technical program committees and an Associate Editor for several IEEE journals.
\end{IEEEbiographynophoto}

\vspace{-110mm}
\begin{IEEEbiographynophoto}
    {Christopher~G.~Brinton} (Senior Member, IEEE) is an Assistant Professor in the Elmore Family School of Electrical and Computer Engineering (ECE) at Purdue University. His research interest is at the intersection of networking, communications, and machine learning, specifically in fog/edge network intelligence, distributed machine learning, and data-driven wireless network optimization. Since joining Purdue ECE in fall 2019, Dr. Brinton has won the NSF CAREER Award (2022), ONR Young Investigator Program (YIP) Award (2022), DARPA Young Faculty Award (YFA, 2022), Intel Rising Star Faculty Award (2022), and roughly \$10M in sponsored research projects as a PI or co-PI. He has also been awarded Purdue College of Engineering Faculty Excellence Awards in Early Career Research (2023), Early Career Teaching (2023), and Online Learning (2022), as well as the Purdue ECE Outstanding Faculty Mentor Award (2020), Ruth and Joel Spira Outstanding Teacher Award (2020), and Purdue Seed for Success Award (2019). He currently serves as an Associate Editor for IEEE Transactions on Wireless Communications, in the ML and AI for wireless area. Prior to joining Purdue, Dr. Brinton was the Associate Director of the EDGE Lab and a Lecturer of Electrical Engineering at Princeton University. He also co-founded Zoomi Inc., a big data startup company that has provided learning optimization to more than one million users worldwide and holds US Patents in machine learning for education. His book The Power of Networks: 6 Principles That Connect our Lives and associated Massive Open Online Courses (MOOCs) have reached over 400,000 students to date. Dr. Brinton received the PhD (with honors) and MS Degrees from Princeton in 2016 and 2013, respectively, both in Electrical Engineering.
\end{IEEEbiographynophoto}

\end{document}